\documentclass[a4paper,USenglish,cleveref,autoref,thm-restate,numberwithinsect,nameinlink]{lipics-v2021}

\newif\ifFinal
\Finaltrue

\usepackage[table]{xcolor} 
\usepackage{tikz}
\usetikzlibrary{arrows.meta}
\usepackage{algorithm2e}
\usepackage[T1]{fontenc}

\newcommand{\Oh}{\ensuremath{\mathcal{O}}}

\newcommand{\proofparagraph}[1]{\medskip\noindent\textit{#1.}}

\DeclareMathOperator{\poly}{poly}
\DeclareMathOperator{\twwithoutN}{{tw}}
\DeclareMathOperator{\vretwithoutN}{{ret}}
\DeclareMathOperator{\eretwithoutN}{{e-ret}}
\DeclareMathOperator{\off}{off}
\DeclareMathOperator{\desc}{desc}
\DeclareMathOperator{\anc}{anc}
\DeclareMathOperator{\DP}{DP}
\newcommand{\Dbar}{{\ensuremath{\overline{D}}}\xspace}
\newcommand{\kbar}{{\ensuremath{\overline{k}}}\xspace}
\newcommand{\w}{{\ensuremath{\omega}}\xspace}
\newcommand\lb{}

\newcommand{\yes}{{\normalfont\texttt{yes}}\xspace}
\newcommand{\no}{{\normalfont\texttt{no}}\xspace}

\newcommand{\XP}{{\normalfont\texttt{XP}}\xspace}
\newcommand{\Wh}[1]{{\normalfont\texttt{{W}}}[\ensuremath{#1}]\xspace}
\newcommand{\NP}{{\normalfont\texttt{NP}}\xspace}
\newcommand{\FPT}{{\normalfont\texttt{FPT}}\xspace}
\newcommand{\SETH}{{\normalfont\texttt{SETH}}\xspace}

\newcommand{\Instance}{{\ensuremath{\mathcal{I}}}\xspace}
\newcommand{\Net}{{\ensuremath{\mathcal{N}}}\xspace}

\newcommand{\vret}{{\ensuremath{\vretwithoutN_{\Net}}}\xspace}
\newcommand{\eret}{{\ensuremath{\eretwithoutN_{\Net}}}\xspace}
\newcommand{\tw}{{\ensuremath{\twwithoutN_{\Net}}}\xspace}
\newcommand{\PD}{\PDsub{\Net}}
\newcommand{\PDsub}[1]{{\ensuremath{PD_{#1}}}\xspace}

\newcommand{\take}{\texttt{take}\xspace}
\newcommand{\leave}{\texttt{leave}\xspace}

\newcommand{\Index}{\texttt{ind}\xspace}

\newcommand{\SstarPrime}[1]{\ensuremath{\mathcal{S}^*_{#1}}\xspace}
\newcommand{\Sstar}{\SstarPrime{t,R,G,B,s}}

\newcommand{\problemdef}[3]{
	\begin{quote}
	  \normalsize\textsc{#1} \smallskip \\
	  \begin{tabularx}{0.9\textwidth}{@{}l@{\hspace{3pt}}X}
		\normalsize\textbf{Input:}	& \normalsize#2 \\
		\normalsize\textbf{Question:} & \normalsize#3
	  \end{tabularx}
	\end{quote}
}

{\upshape\itshape}{\upshape\rmfamily}
\newtheorem{rr}{Reduction~Rule}{\upshape\itshape}{\upshape\rmfamily}

\newcommand{\MAPPD}{{\normalfont\textsc{MapPD}}\xspace}
\newcommand{\MAPPDlong}{{\normalfont\textsc{Max-All-Paths-PD}}\xspace}
\newcommand{\cMAPPD}{{\normalfont\textsc{colored-\MAPPD}}\xspace}
\newcommand{\cMAPPDlong}{{\normalfont\textsc{colored-\MAPPDlong}}\xspace}
\newcommand{\wpSC}{{\normalfont\textsc{wpSC}}\xspace}
\newcommand{\wpSClong}{{\normalfont\textsc{Item-Weighted Partial Set Cover}}\xspace}
\newcommand{\rbnb}{{\normalfont\textsc{Red-Blue Non-Blocker}}\xspace}
\newcommand{\SC}{{\normalfont\textsc{Set Cover}}\xspace}

\newcommand{\KP}{{\normalfont\textsc{Knapsack}}\xspace}

\Crefname{rr}{Rule}{Rules}
\crefname{rr}{Rule}{Rules}
\Crefname{theorem}{Theorem}{Theorems}
\crefname{theorem}{Thm.}{Thms.}
\Crefname{corollary}{Corollary}{Corollaries}
\crefname{corollary}{Cor.}{Cors.}

\ifFinal

\usepackage[disable]{todonotes}
\else

\usepackage[notref, notcite]{showkeys}

\usepackage{todonotes}
\fi

\usepackage{lineno}
\linenumbers

\newcommand{\myrowcols}{\myrowcolsnum{0}}
\newcommand{\myrowcolsnum}[1]{\rowcolors{#1}{gray!10!white}{white}}

\newcommand{\kommentar}[1]{}

\title{Maximizing All-Paths Phylogenetic Diversity: Parameterized Approaches for Networks\footnote{A preliminary version of this paper has been presented at IPEC 2023~\cite{MAPPD}.}} 

\titlerunning{Maximizing All-Paths Phylogenetic Diversity}
 \author{Mark Jones}
 	{Middlesex University, London, United Kingdom}
 	{m.e.l.jones@tudelft.nl}
 	{https://orcid.org/0000-0002-4091-7089}
 	{Partially supported by Netherlands Organisation for Scientific Research~(NWO) grant OCENW.KLEIN.125.}
 \author{Jannik Schestag}
 	{University of Bergen, Bergen, Norway\\
 	The research was largely carried out during an extended research visit of Jannik Schestag at Delft University of Technology, Delft, The Netherlands, in 2023. We thank the German Academic Exchange Service (DAAD), project 57556279, for the support. \and \url{www.JannikSchestag.eu}}
 	{j.t.schestag@uni-jena.de}
 	{https://orcid.org/0000-0001-7767-2970}
 	{Supported by the Dutch Research Council (NWO), grant~OCENW.GROOT.2019.015.}

\authorrunning{Jones and Schestag}



\Copyright{Jones and Schestag}

\hideLIPIcs
\nolinenumbers

\begin{document}

\maketitle

\begin{abstract}
	Phylogenetic Diversity (PD) is a fundamental measure of biodiversity, originally defined on phylogenetic trees and widely used in conservation biology.
	Phylogenetic trees are often generalised to directed acyclic graphs, called phylogenetic networks. As such, a corresponding generalization of PD is needed.
	A natural generalization to edge-weighted phylogenetic networks is the all-paths measure, where the diversity of a set~$S$ of species (taxa) is defined as the total weight of all edges that lie on a path from the root to at least one species in~$S$.
	While maximizing PD on trees can be solved in polynomial time, the corresponding problem on networks is \NP-hard and difficult to approximate.
	We undertake a systematic parameterized complexity study of the \MAPPDlong (\MAPPD) problem.
	We establish W[2]-hardness when parameterized by the number of species that are included in a solution, and W[1]-hardness for the number of species that are excluded.
	On the positive side, we show that the problem is fixed-parameter tractable with respect to the threshold of diversity and the acceptable loss of diversity.
	We further analyze how the network's proximity to a tree influences algorithmic behavior and present single-exponential fixed-parameter algorithms when parameterized by the number of reticulations and by the treewidth of the underlying graph.
	Finally, we present a polynomial kernelization for \MAPPD with respect to the number of reticulation edges.
\end{abstract}


\section{Introduction}
Phylogenetic diversity, first introduced in 1992 by Faith~\cite{FAITH1992}, is a well-established measure of the amount of biodiversity in a set of species. 
It formalizes the intuitive notion that a set of species is likely to have a greater range of biological features when they are distantly related. 
Such a measure is of crucial importance in the field of biological conservation, where there are often insufficient resources available to save every threatened species, one must make hard decisions about which species to prioritize.
Phylogenetic diversity 
forms the basis of the Fair Proportion Index and the Shapley Value~\cite{haake2008shapley,Hartmann2013TheEO,redding2006incorporating}, which are used to evaluate the individual contribution of individual species to overall biodiversity. These measures are used by conservation initiatives such as the IUCN's  Phylogenetic Diversity Task Force (\url{https://www.pdtf .org/})
and the Zoological Society of London's EDGE of Existence program~\cite{LondonEDGE}.

The subject of phylogenetic diversity is the phylogenetic tree; that is, a rooted tree, with weights on the edges, whose leaves represent a subset of present-day species.
The \emph{phylogenetic diversity}~$\PD(S)$ of a set of species~$S$ is the total weight of edge on a (directed) path from the root to one of the leaves in~$S$.
Here the weight of an edge corresponds to phylogenetic distance, which is taken to be proportional to the number of features of interest (e.g. biological characteristics) that emerge along that edge.

Phylogenetic diversity as originally proposed by Faith is defined for phylogenetic trees.
Consequently, it does not allow for models of evolutionary history with reticulation events (where a species inherits genetic data from two or more species), such as hybridization or lateral gene transfer.
Such events are modeled in phylogenetic \emph{networks} (directed acyclic graphs with a single source), which extend the class of phylogenetic trees~\cite{phylogeneticNetworks}.
There are a number of ways to extend phylogenetic diversity to phylogenetic networks~\cite{bordewich,AVGTree,WickeFischer2018}.
In this paper, we consider one of the simplest, \emph{all-paths phylogenetic diversity} (first introduced under the name `phylogenetic subnet diversity' in~\cite{WickeFischer2018} and further studied in~\cite{bordewich}).
Under this measure, given a rooted phylogenetic network~$\Net$ with edge weights, the phylogenetic diversity of~$S$ is again the total weight of all edges on a (directed) path from the root to one of the leaves in~$S$.


Assuming it is not possible to preserve all threatened species (e.g. due to limited resources), we would like to find a subset of species that can be preserved, for which the overall diversity is maximized.
This gives rise to the maximum all-paths phylogenetic diversity problem (\MAPPD): given a network~$\Net$ and integer $k$, find a set of at most $k$ species that maximizes the all-paths phylogenetic diversity score.
Fortunately, in the case of trees, this turns out to be a tractable problem---given as input a phylogenetic tree and number $k$, there is a polynomial-time greedy algorithm that outputs the set of $k$ species with maximum phylogenetic diversity~\cite{Pardi2005,steel}.
Unfortunately, this result does not extend to phylogenetic networks---\MAPPD is \NP-hard, and cannot be approximated in polynomial time with approximation ratio better than $1-\frac{1}{e}$ unless $\texttt{P} = \NP$~\cite{bordewich}.
For this reason, we study the problem from the perspective of parameterized complexity.


%
\paragraph*{Related work}
All-paths phylogenetic diversity as a measure on networks was first introduced in~\cite{WickeFischer2018}.
The computational complexity of \MAPPD was first studied in~\cite{bordewich}, where the authors showed that the problem is NP-hard and cannot be approximated in polynomial time with approximation ratio better than $1-\frac{1}{e}$ unless $\texttt{P} = \NP$, but is polynomial-time solvable on the class of level-1 networks (in which the undirected cycles are pairwise vertex-disjoint).

Phylogenetic diversity forms the basis of the Shapley Value, a measure that describes how much a \emph{single} species contributes to overall biodiversity. The definition of the Shapley Value involves the phylogenetic diversity of every possible subset of species, and so is difficult to calculate directly.
However, it was shown in~\cite{fuchs2015equality} that (on phylogenetic trees) the Shapley Value is equivalent to the Fair Proportion Index~\cite{redding2006incorporating}, which can be calculated in polynomial time.
In the case of phylogenetic networks, it was shown that this result also extends to Shapley Value based on all-paths phylogenetic diversity.
This is in contrast to the NP-hardness result of~\cite{bordewich}---while it is easy to determine the individual species that contributes the most phylogenetic diversity across all sets of species, it is hard to find a \emph{set} of species for which the phylogenetic diversity is maximal.

The phylogenetic networks considered in this paper are \emph{explicit} networks, in which each vertex represents a different species in evolutionary history and the edges represent the transfer of genetic information from one species to another. 
Phylogenetic diversity has also been studied on \emph{split} networks. Such networks do not represent a single explicit evolutionary history, but can represent structural information from several sources (e.g. conflicting phylogenetic trees)---see e.g.~\cite{Chernomor2016,SNM2009,Volkmann2014}.

\paragraph*{Reactions to the previous version of this work}	
The initial publication of this paper in 2023 marked the beginning of a deeper study of phylogenetic diversity on networks.
More realistic models of phylogenetic diversity such as Network-PD, in which each reticulation edge has an inheritance probability~\cite{bordewich,MaxNPD}, and Average-Tree-PD, in which the average PD of all displayed trees is considered~\cite{AVGTree}, have been introduced and analyzed.
As these problem turn out to be computationally even harder,
\MAPPD seems to be a good, tractable core model when analyzing phylogenetic diversity on networks.
Following up on our algorithms, two version of \MAPPD on semi-directed phylogenetic networks have been presented in the software-package \texttt{PaNDA}~\cite{holtgrefe2026tractable,PaNDA}.
Further to our study, \MAPPD has also been analyzed when ecological dependencies are taken into account~\cite{MAPPDD}.
(For more information on phylogenetic diversity with ecological dependencies, see~\cite{faller,1PDDkernel,PDD,moulton,WeightedFW,twvssw}.)

\paragraph*{Our contribution}
We conduct a comprehensive parameterized analysis of the problem \MAPPDlong (\MAPPD), in which the task is to find a set of at most $k$ leaves maximizing the all-paths phylogenetic diversity in a network (see  Section~\ref{sec:prelim} for a formal definition).
Our results delineate a sharp complexity boundary.
We show in Section~\ref{sec:wpSC} that \MAPPD is $\Wh 2$-hard with respect to~$k$, the number of saved species, and \Wh1-hard with resect to $\kbar$, the number of excluded species, by establishing an equivalence between this parameterization of \MAPPD and a  generalization of \textsc{Set Cover} called \wpSClong. 
Both parameterizations are clearly \XP, as it is sufficient to check all sets of size~$k$ as a solution or sets of size~$\kbar$ as the complement of a solution.

On the positive side, we show in Section~\ref{sec:FPTdiversity} that \MAPPD is fixed-parameter tractable (\FPT) with respect to $D$, the total phylogenetic diversity of the desired solution, and also with respect to the `dual'~$\Dbar$, i.e. the acceptable loss in phylogenetic diversity.

Finally we turn to structural parameters. In Section~\ref{sec:FPTtreelike} we give single-exponential fixed-parameter algorithms for \MAPPD with respect to the number of reticulations in the network, and with respect to the treewidth of the underlying graph of the network, two parameters that are small in practice on natural networks that are relatively tree-like.
In the case of reticulations, this algorithm is asymptotically tight under the \texttt{Strong Exponential Time Hypothesis}.
We complete this work by showing in Section~\ref{sec:kernel} that \MAPPD admits a polynomial kernel when parameterized by the number of reticulation edges.

Together, these results provide a detailed parameterized complexity landscape for \MAPPD.
\Cref{tab:results} depicts a summary of the results presented in this paper.

\begin{table}[t]\centering
	\caption{An overview over the parameterized complexity results for \MAPPD.}
	
	\myrowcols
	\begin{tabular}{l ll}
		\hline
		Parameter & \multicolumn{2}{c}{\MAPPDlong} \\
		\hline
		Budget $k$ & \Wh{2}-hard, \XP & \cref{thm:k+maxw} \\
		Diversity $D$ & \FPT & Cor.~\ref{thm:D} \\
		Species-loss $\kbar$ & \Wh{1}-hard, \XP & \cref{thm:kbar} \\
		Diversity-loss $\Dbar$ & \FPT & \cref{thm:Dbar} \\
		\hline
		Number of & \FPT & \cref{thm:ret} \\
		\rowcolor{white}
		reticulations~\vret & poly-kernel \emph{open} &  \\
		\rowcolor{gray!10!white}
		Number of & \FPT & \cref{thm:ret} \\
		ret.-edges~\eret & poly-kernel & \cref{thm:kernel} \\
		Treewidth & \FPT & \cref{thm:tw} \\
		\hline
	\end{tabular}
	\label{tab:results}
\end{table}

\section{Preliminaries}\label{sec:prelim}
For an integer~$\ell$, by~$[\ell]$ we denote the set~$\{1,\dots,\ell\}$ and $[\ell]_0:=\{0\}\cup [\ell]$.

In our running time analyses, we assume a unit-cost RAM model where arithmetic addition of integers~$N$ and~$M$ can be done in~$\Oh(\log N + \log M)$~time.
We use~$\Oh^*$ to denote running time complexity ignoring polynomial factors.

\paragraph*{Phylogenetic Networks and Phylogenetic Diversity}
A  \emph{phylogenetic $X$-network~$\Net=(V,E,\w)$} is a directed, acyclic graph~$G = (V,E)$ with \emph{edge-weight} function~$\w: E\to \mathbb{N}_{>0}$ and a single vertex of in-degree $0$ (the \emph{root}), in which the vertices of out-degree $0$ (the \emph{leaves}) have in-degree $1$ and are bijectively labeled with elements from a set~$X$, and such that all vertices have in-degree at most $1$ or out-degree at most $1$.
For brevity, we usually refer to a phylogenetic $X$-network as an \emph{$X$-network}, or more simply a \emph{network} when the set $X$ is not relevant.
The vertices with in-degree at least $2$ and out-degree $1$ are called \emph{reticulations}; the other non-leaf vertices are called \emph{tree vertices}.
Any edge incoming at a reticulation is a \emph{reticulation edge}.

In biological applications, the set $X$ is a set of \emph{taxa} (species), the internal vertices of~$\Net$ correspond to biological ancestors of these taxa and~$\w(e)$ describes the phylogenetic distance between the endpoints of~$e$ (as these endpoints correspond to distinct species, we may assume this distance is greater than $0$).
Throughout the paper we use the conventions that $n:= |V|$, $m:= |E|$, and~$\rho$ is the root of the network.

For a vertex $v$, the set of~\emph{descendants $\desc(v)$ of~$v$} is the set of vertices $u$ for which there is a path from $v$ to $u$ (note that $v \in \desc(v)$).
We analogously define the set of~\emph{ancestors $\anc(v)$ of~$v$} as the set of vertices $u$ for which there is a path from $u$ to $v$.
The~\emph{offspring~$\off(v)$ of~$v$} is the intersection of $\desc(v)$ and $X$.
Further, for each edge $e=(v,w)$, we define $\anc(e)=\anc(v)$, $\desc(e)=\desc(w)$ and $\off(e)=\off(w)$.
For each set of taxa $Y$, an edge $e$ is \emph{affected by $Y$} if $\off(e)\cap Y \ne \emptyset$ and \emph{strictly affected by $Y$} if $\off(e)\subseteq Y$.
The sets $T_Y$ and $E_Y$ are the sets of edges strictly affected and affected  by $Y$, respectively.
For each set of taxa $Y$, the \emph{all-paths phylogenetic diversity $\PD(Y)$ of $Y$} is

$$
\PD(Y) := \sum_{e\in E_Y} \w(e).
$$

That is, $\PD(Y)$ is the total weight of all edges $(u,v)$ in $\Net$ such that there is a path from $v$ to a vertex in $Y$. 
We refer to $\PD(Y)$ simply as the \emph{phylogenetic diversity of $Y$}.

\paragraph*{Parameterized Complexity and Kernelizations}
For a detailed introduction to parameterized complexity refer to the standard monographs~\cite{cygan,downeybook}.
We only present a small introduction.

There is a natural bijection between decision problems and formal languages over some alphabet $\Sigma$.
A \emph{parameterized language} is subset $L \subseteq \Sigma^* \times \mathbb{N}_0$.
For each pair $(\sigma,k) \in L$, $k$ is called the \emph{parameter} of this pair.
A parameterized language is \emph{fixed-parameter tractable} (\FPT) if there exists an algorithm $\mathcal{A}$ that decides for any input $(\sigma,k) \in \Sigma^* \times \mathbb{N}_0$ whether $(\sigma,k) \in L$ in time $f(k) \cdot \poly(|\sigma|)$, for some $f(\cdot)$ computable function.
A parameterized language is \emph{slicewise polynomial} (\XP) if there exists an algorithm $\mathcal{A}$ that decides for any input $(\sigma,k) \in \Sigma^* \times \mathbb{N}_0$ whether $(\sigma,k) \in L$ in time $|\sigma|^{f(k)}$, for some computable function~$f(\cdot)$.
Every problem that is \FPT is also in \XP.
It is widely believed that $\Wh{h}$-hard problems, for~$h \ge 1$, are not in \FPT.

A \emph{kernel}---also called \emph{kernelization} or \emph{kernelization algorithm}---for a parameterized language $\Pi \subseteq \Sigma^* \times \mathbb{N}_0$ is an algorithm $\mathcal{A}$ which, given a pair $(\sigma,k) \in \Sigma^* \times \mathbb{N}_0$, constructs in~$\poly(|\sigma|+k)$~time another pair $(\tau,\ell) \in \Sigma^* \times \mathbb{N}_0$ such that $\ell \le f(k)$, and $(\sigma,k) \in \Pi$ if and only if $(\tau,\ell) \in \Pi$,
where $f(\cdot)$ is some computable function.
A kernel $\mathcal{A}$ is \emph{polynomial} if $f(\cdot)$ is a polynomial.

\paragraph*{Problem Definitions and Parameterizations}
Our main object of study is the following problem, introduced in~\cite{bordewich}:

\problemdef{\MAPPDlong (\MAPPD)}
{A phylogenetic $X$-network \Net and two integers $k$ and $D$.}
{Is there a subset $Y\subseteq X$ of taxa with size at most $k$ and phylogenetic diversity at least $D$? That is, $|Y|\le k$ and $\PD(Y)\ge D$.}

When $Y$ satisfies the conditions above, we say $Y$ is a \emph{solution} for \MAPPD on $(\Net,k,D)$.

In Section \ref{sec:wpSC}, we show that there is a strong connection between \MAPPD and the problem \wpSClong, which is defined as follows.
\problemdef{\wpSClong (\wpSC)}
{A universe $\mathcal{U}$, a family $\mathcal{F}$ of subsets over $\mathcal{U}$, an integer weight $\w(u)$ for each item $u\in\mathcal U$ and two integers $k$ and $D$.}
{Are there sets $F_1,\dots,F_k\in \mathcal F$ such that sum of the weights of the elements in $L := \bigcup_{i=1}^k F_i$ is at least $D$?
That is $\sum_{u\in L} \w(u) \ge D$.
}
When $F_{1},\dots,F_{\ell}$ satisfy the above condition, then $\{F_{1},\dots,F_{\ell}\}$ is a \emph{solution} for \wpSC on~$(\mathcal{U}, \mathcal{F},\w, k, D)$.

\textsc{Set Cover} is the special case of \wpSC with $D=|\mathcal{U}|$ and $\w(u)=1$ for each $u\in\mathcal U$.

%
In addition to the parameters $k$ and $D$, which are the number of saved taxa and the preserved phylogenetic diversity,
we also study the `dual' parameters which are the minimum number of species that will go extinct $\kbar := |X|-k$ and the acceptable loss of phylogenetic diversity $\Dbar := \PD(X) - D$.
By $\max_\w$, we denote the biggest weight of any edge.
By \vret, we denote the number of reticulations in \Net
and, by  \eret, we denote the number of reticulation-edges which need to be removed such that~\Net is a tree.
More formally, $\eret := \sum_{r \in R} (\deg^-(r) -1) = |E| - |V| + 1$, where~$R$ is the set of reticulations of~\Net and~$\deg^-(v)$ is the in-degree of a vertex~$v$.
Some authors refer to~\eret as the reticulation number, since ~$\eret = \vret$ on binary networks.
By \tw, we denote the treewidth of the underlying undirected graph of \Net.
We formally define treewidth in \Cref{def:tw}.
For further insight into treewidth, we refer to~\cite[Chapter 7]{cygan}.

A phylogenetic $X$-network is called \emph{binary} if the root has a degree of~2 and each other non-leaf has a degree of~3.
Throughout this work, we do \emph{not} assume the network to be binary.
In a previous version of this paper~\cite{MAPPD}, we mistakenly stated that it is possible to transform phylogenetic networks in polynomial time to binary phylogenetic networks without increasing the treewidth.
The construction was false, but the remaining algorithms are correct.

\section{Relationship to \wpSClong}
\label{sec:wpSC}
In this section, we demonstrate a relationship between \MAPPD and \wpSC by presenting reductions in both directions.
Bordewich et al. already proved a similar reduction from \SC to \MAPPD \cite{bordewich}.
\begin{lemma}
	\label{lem:PSC->MAPPD}
	For every instance $\Instance = (\mathcal{U},\mathcal{F},\w,k,D)$ of \wpSC, 
	\begin{enumerate}
		\item\label{it:P->Mweighted} an equivalent instance~$\Instance' =(\Net,k',D')$ of \MAPPD with $k'=k$ and $|X|=\vret=|\mathcal F|$, can be computed in time polynomial in $|\mathcal U|+|\mathcal F|$;
		\item\label{it:P->Munweighted} an equivalent instance~$\Instance_2' = (\Net = (V,E,\w'),k',D')$ of \MAPPD in which $k'=k$ and each edge has weight 1, can be computed in time polynomial in $|\mathcal U|+|\mathcal F|+\max_\w$.
	\end{enumerate}

\end{lemma}
This lemma has several implications for the complexity of \MAPPD.
Because \SC is \Wh 2-hard with respect to the size of the solution $k$, \wpSC and consequently \MAPPD are as well. 
\begin{theorem}
	\label{thm:k+maxw}
	\MAPPD is $\Wh 2$-hard when parameterized with $k$, even if $\max_\w=1$.
\end{theorem}

In \rbnb an undirected bipartite graph $G$ with vertex bipartition $V=V_r \cup V_b$ and an integer $k$ are given.
The question is whether there is a set $S\subseteq V_r$ of size at least $k$ such that each vertex $v$ of $V_b$ has a neighbor in $V_r \setminus S$.
%
There is a standard reduction from \rbnb to \SC:
Let $V_b$ be the universe, for each vertex $v\in V_r$ add a set $F_v:=N(v)$ to $\mathcal F$ and finally set $k':=|V_r|-k$.
\rbnb is \Wh 1-hard when parameterized by the size of the solution~\cite{downey}.
Hence, \SC is \Wh 1-hard with respect to $|\mathcal F|-k$ and with \Cref{lem:PSC->MAPPD} we conclude as follows.
\begin{theorem}
	\label{thm:kbar}
	\MAPPD is $\Wh 1$-hard when parameterized with $\kbar=|X|-k$.
\end{theorem}

\MAPPD can be solved in $\Oh^*(2^{|X|})$ with a brute force algorithm that tries every possible subset of species as a solution.
In \Cref{thm:ret}, we prove that \MAPPD can be solved in~$\Oh^*(2^{\vret})$~time.
In order to prove that these algorithms can not be improved significantly, we apply the well-established \texttt{Strong Exponential Time Hypothesis} (\SETH) \cite{ Calabro2009, impagliazzo2001complexity}.

Unless \SETH fails, \SC can not be solved in $\Oh^*(2^{\epsilon \cdot |F|})$ time for any $\epsilon<1$ \cite{cyganSETH,lin}.
By \Cref{lem:PSC->MAPPD}, this implies \MAPPD can not be solved in time $\Oh^*(2^{\epsilon \cdot |X|})$ or $\Oh^*(2^{\epsilon \cdot \vret})$ for any $\epsilon  <1$.
Therefore, the $\Oh^*(2^{|X|})$ and $\Oh^*(2^{\vret})$ time algorithms described above are tight under \SETH.

\begin{corollary}
	\label{cor:X}
	\MAPPD can not be solved in $\Oh(2^{\epsilon\cdot |X|}) \cdot \poly(|\Instance|)$ time or in $\Oh(2^{\epsilon\cdot \vret}) \cdot \poly(|\Instance|)$ time for any $\epsilon < 1$, unless \SETH fails.
\end{corollary}

\noindent So now, without further ado, we prove \Cref{lem:PSC->MAPPD}.

\begin{proof}[Proof of \Cref{lem:PSC->MAPPD}]
	\proofparagraph{Reduction}
	Let $\Instance = (\mathcal{U},\mathcal{F},k,D)$ be an instance of \wpSC. Let $\mathcal{U}$ consist of the items $u_1,\dots,u_s$ and let $\mathcal{F}$ contain the sets $F_1,\dots,F_t$.
	We may assume that for each $u_i$ there is a set $F_j$ which contains $u_i$.
	We define an instance $\Instance' = (\Net,k,D')$ of \MAPPD as follows.
	Let $k' = k$ and define $D' := D\cdot Q + 1$ for $Q := t(s+1)$.
	We define a network~\Net with leaves $x_1,\dots,x_t$, and additional vertices $\rho,v_1,\dots,v_s,w_1,\dots,w_t$.

	Let the set of edges in~\Net consist of the edges $(\rho,v_{i})$ for $i\in [s]$, $(w_j,x_j)$ for $j\in [t]$, and let $(v_i,w_j)$ be an edge if and only if $u_i \in F_j$.
	We define the weight of $(\rho,v_i)$ to be $\w(u_i) \cdot Q$ for each $i\in [s]$ and 1 for each other edge.
	\Cref{fig:PSC->MAPPD} depicts an example of this reduction.
	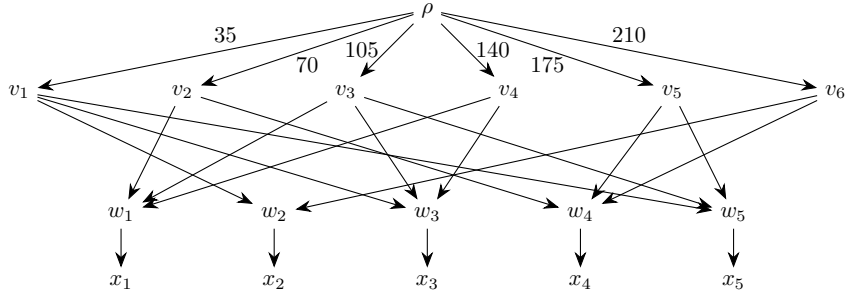
\begin{figure}[t]
		\centering
		\begin{tikzpicture}[scale=0.9,every node/.style={scale=0.85}]
			\node (root) at (0,3) {$\rho$};
			\begin{scope}[name prefix = u]
				\node (a) at (-6,1.8) {$v_1$};
				\node (b) at (-3.6,1.8) {$v_2$};
				\node (c) at (-1.2,1.8) {$v_3$};
				\node (d) at (1.2,1.8) {$v_4$};
				\node (e) at (3.6,1.8) {$v_5$};
				\node (f) at (6,1.8) {$v_6$};
			\end{scope}
			\begin{scope}[name prefix = F]
				\node (a) at (-4.5,0) {$w_1$};
				\node (b) at (-2.25,0) {$w_2$};
				\node (c) at (0,0) {$w_3$};
				\node (d) at (2.25,0) {$w_4$};
				\node (e) at (4.5,0) {$w_5$};
			\end{scope}
			\begin{scope}[name prefix = x]
				\node (a) at (-4.5,-1) {$x_1$};
				\node (b) at (-2.25,-1) {$x_2$};
				\node (c) at (0,-1) {$x_3$};
				\node (d) at (2.25,-1) {$x_4$};
				\node (e) at (4.5,-1) {$x_5$};
			\end{scope}
			
			\draw[-{Stealth[length=6pt]}] (root) to node[above]{35} (ua);
			\draw[-{Stealth[length=6pt]}] (root) to node[below]{70} (ub);
			\draw[-{Stealth[length=6pt]}] (root) to node[left]{105} (uc);
			\draw[-{Stealth[length=6pt]}] (root) to node[right]{140} (ud);
			\draw[-{Stealth[length=6pt]}] (root) to node[below]{175} (ue);
			\draw[-{Stealth[length=6pt]}] (root) to node[above]{210} (uf);

			\draw[-{Stealth[length=6pt]}] (ub) to (Fa);
			\draw[-{Stealth[length=6pt]}] (uc) to (Fa);
			\draw[-{Stealth[length=6pt]}] (ud) to (Fa);
			\draw[-{Stealth[length=6pt]}] (ua) to (Fb);
			\draw[-{Stealth[length=6pt]}] (uf) to (Fb);
			\draw[-{Stealth[length=6pt]}] (ua) to (Fc);
			\draw[-{Stealth[length=6pt]}] (uc) to (Fc);
			\draw[-{Stealth[length=6pt]}] (ud) to (Fc);
			\draw[-{Stealth[length=6pt]}] (ub) to (Fd);
			\draw[-{Stealth[length=6pt]}] (ue) to (Fd);
			\draw[-{Stealth[length=6pt]}] (uf) to (Fd);
			\draw[-{Stealth[length=6pt]}] (ua) to (Fe);
			\draw[-{Stealth[length=6pt]}] (uc) to (Fe);
			\draw[-{Stealth[length=6pt]}] (ue) to (Fe);

			\draw[-{Stealth[length=6pt]}] (Fa) to (xa);
			\draw[-{Stealth[length=6pt]}] (Fb) to (xb);
			\draw[-{Stealth[length=6pt]}] (Fc) to (xc);
			\draw[-{Stealth[length=6pt]}] (Fd) to (xd);
			\draw[-{Stealth[length=6pt]}] (Fe) to (xe);
		\end{tikzpicture}
		\caption{This figure depicts the network \Net that we reduce to from the instance\lb $(\mathcal{U}:=\{u_1,\dots,u_6\},\mathcal{F}:=\{F_1,\dots,F_5\},\w,k,D)$ of \wpSC with $\w(u_i)=i$, $F_1:=\{u_2,u_3,u_4\}$, $F_2:=\{u_1,u_6\}$, $F_3:=\{u_1,u_3,u_4\}$, $F_4:=\{u_2,u_5,u_6\}$, $F_5:=\{u_1,u_3,u_5\}$.
			Unlabeled edges have a weight of 1.
			Here $s=6,t=5$ and $Q = 35$.
			The value of $k'$ would be $k$ and $D'$ would be $35D + 1$.}
		\label{fig:PSC->MAPPD}
	\end{figure}%
	
	This completes the construction of instance $\Instance'$ in Case~\ref{it:P->Mweighted} of the theorem.
	We now describe how to construct an instance $\Instance_2'$ from $\Instance'$ in which the weight of each edge is~1, completing the construction for Case~\ref{it:P->Munweighted}.
	For each edge $e=(\rho,v_i)$ with $w(e) > 1$, make $\w(e)-1$ subdivisions and attach a new leaf as the child of each subdividing vertex. We call these newly-added leaves \emph{false leaves}, and we call the other leaves of~\Net \emph{true leaves}.

	\proofparagraph{Correctness}
	The instance~$\Instance'$ is computed in time polynomial in $|\mathcal U|+|\mathcal F|$ and
	the instance~$\Instance_2'$ is computed in time polynomial in $|\mathcal U|+|\mathcal F|+\max_\w$.
	Clearly, in $\Instance'$ we observe $k'=k$ and $|X|=\vret=|\mathcal F|$ and
	in $\Instance_2'$ we observe $k'=k$ and $\max_{\omega'}=1$.
	It remains to show the equivalence of the instances.

	Without loss of generality, let $S:=\{F_{1},\dots,F_{\ell}\}$ with $\ell\le k$ be a solution for \wpSC on the instance \Instance that covers the items $u_1,\dots,u_p$.
	That is, we assume $S$ consists of the first~$\ell$ sets in $\mathcal F$ and they cover the first $p$ elements in $\mathcal U$.
	We show that $Y:=\{x_{1},\dots,x_{\ell}\}$ is a solution for \MAPPD on the instances $\Instance'$ and $\Instance_2'$.
	Clearly, the size of $Y$ is at most $k'$.
	Now consider the phylogenetic diversity of $Y$ in the network of $\Instance'$.
	Let $\hat E$ be the edges in \Net between two vertices of $v_1,\dots,v_p,w_1,\dots,w_\ell$.
	Then
	\begin{eqnarray*}
		\PD(Y) &=&  \sum_{i=1}^\ell \w'((w_i,x_i)) + \sum_{e\in \hat E} \w'(e) + \sum_{i=1}^p \w'((\rho,v_i))
		\ge \w'((w_1,x_1)) + \sum_{i=1}^p \w'((\rho,v_i))\\
		&=& 1 + \sum_{i=1}^p \w(u_i)\cdot Q
		= 1 + Q \cdot \sum_{i=1}^p \w(u_i)
		\ge 1 + QD = D'.
	\end{eqnarray*}
	Here, the first equality holds because the sets in $S$ cover the items $u_1,\dots,u_p$ and the last inequality holds because $S$ is a solution.
	It is easy to see that the phylogenetic diversity of the set $Y$ in the network of $\Instance_2'$ is identical.
	Hence, $Y$ is a solution for \MAPPD on $\Instance'$ and $\Instance_2'$.
	
	We now show that the existence of a solution for \MAPPD on $\Instance'$ or $\Instance_2'$ implies the existence of a solution for \wpSC on $\Instance$.
	Without loss of generality, let $Y:=\{x_1,\dots,x_\ell\}$ with $\ell\le k$ be a solution for \MAPPD on instance $\Instance'$.
	We show that $S:=\{F_1,\dots,F_\ell\}$ is a solution for \wpSC on instance~\Instance.
	Clearly, the size of $S$ is at most $k$.
	Without loss of generality, let $v_1,\dots,v_p$ be the ancestors of $Y$ that are children of $\rho$.
	Thus, there is a path from $v_i$ to a taxon $x_j\in Y$, for each $i\in [p]$.
	By the construction of $\Instance'$, we conclude~$u_i \in F_j$.
	Again, let $\hat E$ be the edges between two vertices of $v_1,\dots,v_p,w_1,\dots,w_\ell$. We have $D' \le \PD(Y) = \sum_{i=1}^\ell \w'((w_i,x_i)) + \sum_{e\in \hat E} \w'(e) + \sum_{i=1}^p \w'((\rho,v_i)) \le t + st + \sum_{i=1}^p \w'((\rho,v_i))$.
	Consequently, $\sum_{i=1}^p Q\cdot \w(u_i) = \sum_{i=1}^p \w'((\rho,v_i)) \ge D' - Q = Q(D-1) + 1$.
	We conclude that $\sum_{i=1}^p \w(u_i) \ge D-1 + 1/Q$.
	Since the weights of $u_i$ are integers, it follows that $\sum_{i=1}^p \w(u_i) \ge D$.

	It remains to show that for each solution $Y$ for \MAPPD on instance $\Instance_2'$, an equivalent solution for \MAPPD on instance $\Instance'$ exists.
	If $Y$ does not contain false leaves, then 
	as previously observed, the phylogenetic diversity of $Y$ is the same in $\Instance'$ as in $\Instance_2'$.
	%
	%
	Assume now otherwise and let $z$ be a false leaf in $Y$ and let $p_z$ be the parent of $z$.
	We consider two different cases, $p_z$ has an offspring in $Y\setminus\{z\}$, or not.
	In the former case, we observe $\PD(Y) = \PD(Y\setminus\{z\}) + \w((p_z,z)) = \PD(Y\setminus\{z\}) + 1$.
	Consequently, we can replace~$z$ with any true leaf that is not yet in $Y$ to receive another solution for \MAPPD on $\Instance_2'$.
	In the second case, let the true leaf $x_i$ be an offspring of $p_z$.
	Because of the assumption, $x_i\not\in Y$.
	Then, $\PD(Y) - \w((p_z,z)) \le \PD((Y\setminus \{z\}) \cup \{x_i\}) - \w((w_i,x_i))$.
	Consequently, $Y\setminus \{z\} \cup \{x_i\}$ is also a solution.
	Therefore, we can replace all false leaves and we are done.
\end{proof}

Note that the network constructed in the proof of \Cref{lem:PSC->MAPPD} has one layer of tree vertices and one layer of reticulations.
It might appear that by adding more layers of reticulations and tree vertices to the construction of $\Net$, one could reduce from problems even more complex than \wpSC, and thereby show that \MAPPD has an even higher position in the \texttt{W}-hierarchy.
%
This however is unlikely, because of the reduction to \wpSC that we are about to show.

\begin{lemma} \label{lem:MAPPD->PSC}
	For every instance $\Instance=(\Net,k,D)$ of \MAPPD, we can compute an equivalent instance~$(\mathcal U,\mathcal F,\w,k',D')$ of \wpSC with $k'=k$, $D'=D$ and $\max_{\w'}=\max_{\w}$ in time polynomial in $|\Instance|$.
\end{lemma}
\begin{proof}
	\proofparagraph{Reduction}
	Let $\Instance = (\Net,k,D)$ be an instance of \MAPPD.
	We define an instance $\Instance' = (\mathcal{U},\mathcal{F},\w',k',D')$ of \wpSC as follows.
	Let $k' := k$ and $D':= D$.
	For each edge $e$ of \Net, define an item $u_e$ with weight $\w'(u_e)=\w(e)$ and let $\mathcal{U}$ be the set of these $u_e$.
	For each taxon~$x$, define a set $F_x$ which contains item $u_e$ if and only if $e$ is affected by $\{x\}$.
	Let $\mathcal{F}$ be the family of these $F_x$.

	\proofparagraph{Correctness}
	Clearly, the reduction is computed in polynomial time.
	We show the equivalence of the two instances.
	
	Let $Y$ be a solution for \MAPPD on instance \Instance.
	Without loss of generality, assume $Y=\{x_1,\dots,x_\ell\}$ with $\ell\le k$.
	We show that $F_1,\dots,F_\ell$ is a solution for \wpSC on $\Instance'$.
	By definition, $\ell\le k'$.
	Let $E_Y$ be the edges affected by $Y$.
	Observe that $e$ is in $E_Y$ if and only if $u_e$ is in $F^+ := \bigcup_{i=1}^\ell F_i$.
	Then, $D \le \PD(Y) = \sum_{e\in E_Y} \w(e) = \sum_{u_e\in F^+} \w'(u_e)$.
	Hence, $F_1,\dots,F_\ell$ is a solution for \wpSC on $\Instance'$.

	Now, without loss of generality, let~$F_1,\dots,F_\ell$ be a solution for \wpSC on $\Instance'$.
	Let $u_{e_1},\dots,u_{e_p}$ be the items in the union of $F_1,\dots,F_\ell$.
	By the construction, the edges $e_1,\dots,e_p$ are affected by $Y=\{x_1,\dots,x_\ell\}$.
	Then, $\PD(Y)\ge \sum_{i=1}^p \w(e_i) = \sum_{i=1}^p \w'(u_{e_i}) \ge D$.
	Because the size of $Y$ is at most $k'$, $Y$ is a solution for \MAPPD on $\Instance$.
\end{proof}

To the best of our knowledge, it is unknown if \wpSC is $\Wh 2$-complete, like \textsc{Set Cover}.
Nevertheless, we obtain the following connection between \wpSC and \MAPPD.

\begin{corollary}
	\label{cor:k}
	\MAPPD is $\Wh t$-complete with respect to $k$ if and only if \wpSC is $\Wh t$-complete with respect to $k$.
\end{corollary}

\section{Fixed-Parameter Tractability Results}
\label{sec:FPT}

\subsection{Preserved and lost Diversity}
\label{sec:FPTdiversity}
In this subsection, we show that \MAPPD is \FPT with respect to the threshold of phylogenetic diversity $D$ and the acceptable loss of phylogenetic diversity $\Dbar := \PD(X)-D$.

Let \Instance be an instance of \MAPPD. If there is an edge $e$ with $\w(e)\ge D$ and $k\ge 1$,  then for each offspring $x$ of $e$ we have $\PD(\{x\})\ge\w(e)\ge D$, and so $\{x\}$ is a solution for \MAPPD on \Instance.
So, we may assume that $\max_\w < D$.
Therefore, each edge $e$ can be subdivided $\w(e)-1$ times in $\Oh(D\cdot m)$ time such that $\w'(e)=1$ for each edge $e$ of the new  network $\Net'$.
Bläser showed that \wpSC can be solved in $\Oh^*(2^{\Oh(D)})$ time when $\w(u)=1$ for each item $u\in\mathcal{U}$~\cite[Theorem 2]{blaeser}.
Consequently, with \Cref{lem:MAPPD->PSC} and the result from Bläser we conclude the following.
\begin{theorem}
	\label{thm:D}
	\MAPPD can be solved in $\Oh^*(2^{\Oh(D)})$ time.
\end{theorem}
As \SC is a special case of \wpSC with $D = \sum_{u\in \mathcal{U}}\w(u)$, \wpSC is \NP-hard even if the dual $\sum_{u\in \mathcal{U}}\w(u) - D$ has a size of~0.
By contrast, we show in the following that \MAPPD is \FPT with respect to \Dbar.
More precisely, we show the following.

\begin{theorem}
	\label{thm:Dbar}
	\MAPPD can be solved in $\Oh(2^{\Dbar + \kbar + o(\Dbar)} \cdot m\log n)$~time.
\end{theorem}

Observe that we may assume~$\Dbar \ge \kbar := |X| - k$, as otherwise no solution exists, because we require each edge to have a positive weight.

To show \Cref{thm:Dbar}, we use the technique of color coding.
Recall that $\off(e):=\off(w)$ for each edge $e=(v,w)$.
%
We define the following auxiliary problem,
in which we assign a \emph{color}, red or green, to each taxon.
A set~$Y\subseteq X$ is \emph{color-fitting} if each taxon $x\in Y$ is red and for each vertex~$v \in V(\Net)$, at least one of the following holds:
\begin{itemize}
	\item $v$ has a green offspring,
	\item all offspring of $v$ are in $Y$, or
	\item all offspring of $v$ are in $X \setminus Y$.
\end{itemize}

Effectively, this condition ensures that if we delete all the edges with green offspring, the resulting graph has a set of connected components where for each component, either all the taxa in the component are in $Y$, or none of them are. 
We will make use of this fact when considering the auxiliary problem below.

\problemdef{\cMAPPDlong (\cMAPPD)}
{A phylogenetic $X$-network \Net, integers $k$ and $D$, and a coloring on the taxa $c: X\to \{\text{red},\text{green}\}$.}
{Is there a subset $S\subseteq X$ of taxa such that $|S|\le k$, $\PD(S)\ge D$, and~$X\setminus S$ is color-fitting?}

\begin{lemma} \label{lem:Dbar}
	\cMAPPD can be solved in~$\Oh(\Dbar \cdot m \cdot \log(\Dbar))$~time.
\end{lemma}
\begin{proof}
	\proofparagraph{Algorithm}
	Let $\Instance := (\Net:=(V,E,\w),k,D,c)$ be an instance of \cMAPPD.
	Delete all edges~$(u,v)$ for which~$v$ has a green offspring, and then delete any isolated vertices.
	For any vertex~$u$ of in-degree~$0$ with children~$v_1,\dots, v_q$, replace~$u$ with~$q$ vertices~$u_1,\dots, u_q$ and add and edge~$(u_i,v_i)$ of weight~$\w((u,v_i))$ for each~$i \in [q]$, so that each in-degree-$0$ vertex now has one child.
	Let $G'$ be the resulting graph.
	An example of this transformation is depicted in~\Cref{fig:transformation}.
	
	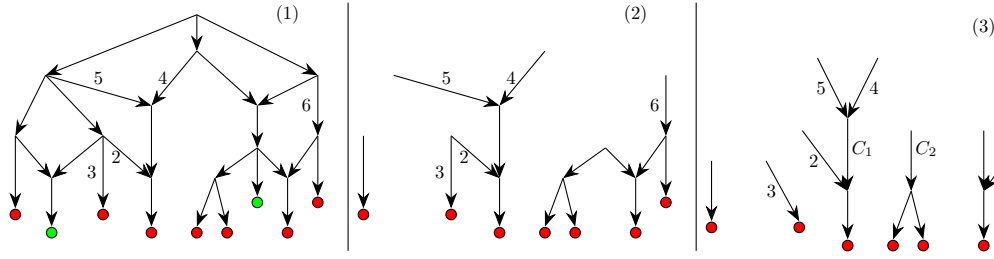
\begin{figure}[t]
		\centering
		\begin{tikzpicture}[scale=0.8,every node/.style={scale=0.7}]
			\draw[-{Stealth[length=6pt]}] (5,8) -> (7,7);
			\draw[-{Stealth[length=6pt]}] (5,8) -> (5,7.4);
			\draw[-{Stealth[length=6pt]}] (5,8) -> (2.5,7);
			
			\draw[-{Stealth[length=6pt]}] (2.5,7) -> (2,6);
			\draw[-{Stealth[length=6pt]}] (2.5,7) -> (3.45,6);
			\draw[-{Stealth[length=6pt]}] (2.5,7) -> node[above] {5} (4.25,6.5);
			
			\draw[-{Stealth[length=6pt]}] (5,7.4) -> node[left] {4} (4.25,6.5);
			\draw[-{Stealth[length=6pt]}] (5,7.4) -> (6,6.5);
			
			\draw[-{Stealth[length=6pt]}] (7,7) -> (6,6.5);
			\draw[-{Stealth[length=6pt]}] (7,7) -> node[left] {6} (7,6);
			
			\draw[-{Stealth[length=6pt]}] (6,6.5) -> (6,5.8);
			\draw[-{Stealth[length=6pt]}] (4.25,6.5) -> (4.25,5.3);

			\draw[-{Stealth[length=6pt]}] (3.45,6) -> node[left] {2} (4.25,5.3);
			\draw[-{Stealth[length=6pt]}] (3.45,6) -> node[left] {3} (3.45,4.8);
			\draw[-{Stealth[length=6pt]}] (3.45,6) -> (2.6,5.3);
			\draw[-{Stealth[length=6pt]}] (2,6) -> (2.6,5.3);
			\draw[-{Stealth[length=6pt]}] (2,6) -> (2,4.8);
			\draw[-{Stealth[length=6pt]}] (4.25,5.3) -> (4.25,4.5);
			\draw[-{Stealth[length=6pt]}] (2.6,5.3) -> (2.6,4.5);

			\draw[-{Stealth[length=6pt]}] (6,5.8) -> (5.3,5.3);
			\draw[-{Stealth[length=6pt]}] (6,5.8) -> (6,5);
			\draw[-{Stealth[length=6pt]}] (6,5.8) -> (6.5,5.3);
			\draw[-{Stealth[length=6pt]}] (7,6) -> (6.5,5.3);
			\draw[-{Stealth[length=6pt]}] (7,6) -> (7,5);
			
			\draw[-{Stealth[length=6pt]}] (6.5,5.3) -> (6.5,4.5);
			\draw[-{Stealth[length=6pt]}] (5.3,5.3) -> (5,4.5);
			\draw[-{Stealth[length=6pt]}] (5.3,5.3) -> (5.5,4.5);
			
			\node[draw,fill=red,inner sep=2pt,circle] at (2,4.7) {};
			\node[draw,fill=green,inner sep=2pt,circle] at (2.6,4.4) {};
			\node[draw,fill=red,inner sep=2pt,circle] at (3.45,4.7) {};
			\node[draw,fill=red,inner sep=2pt,circle] at (4.25,4.4) {};
			
			\node[draw,fill=red,inner sep=2pt,circle] at (5,4.4) {};
			\node[draw,fill=red,inner sep=2pt,circle] at (5.5,4.4) {};
			\node[draw,fill=green,inner sep=2pt,circle] at (6,4.9) {};
			\node[draw,fill=red,inner sep=2pt,circle] at (6.5,4.4) {};
			\node[draw,fill=red,inner sep=2pt,circle] at (7,4.9) {};
			
			\node at (6.5,8) {(1)};
			
			\path[draw] (7.5,8.2) -> (7.5,4.1);
		\end{tikzpicture}
		\begin{tikzpicture}[scale=0.8,every node/.style={scale=0.7}]
			\draw[-{Stealth[length=6pt]}] (2.5,7) -> node[above] {5} (4.25,6.5);
			
			\draw[-{Stealth[length=6pt]}] (5,7.4) -> node[left] {4} (4.25,6.5);
			
			\draw[-{Stealth[length=6pt]}] (7,7) -> node[left] {6} (7,6);
			
			\draw[-{Stealth[length=6pt]}] (4.25,6.5) -> (4.25,5.3);

			\draw[-{Stealth[length=6pt]}] (3.45,6) -> node[left] {2} (4.25,5.3);
			\draw[-{Stealth[length=6pt]}] (3.45,6) -> node[left] {3} (3.45,4.8);
			\draw[-{Stealth[length=6pt]}] (2,6) -> (2,4.8);
			\draw[-{Stealth[length=6pt]}] (4.25,5.3) -> (4.25,4.5);

			\draw[-{Stealth[length=6pt]}] (6,5.8) -> (5.3,5.3);
			\draw[-{Stealth[length=6pt]}] (6,5.8) -> (6.5,5.3);
			\draw[-{Stealth[length=6pt]}] (7,6) -> (6.5,5.3);
			\draw[-{Stealth[length=6pt]}] (7,6) -> (7,5);
			
			\draw[-{Stealth[length=6pt]}] (6.5,5.3) -> (6.5,4.5);
			\draw[-{Stealth[length=6pt]}] (5.3,5.3) -> (5,4.5);
			\draw[-{Stealth[length=6pt]}] (5.3,5.3) -> (5.5,4.5);
			
			\node[draw,fill=red,inner sep=2pt,circle] at (2,4.7) {};
			\node[draw,fill=red,inner sep=2pt,circle] at (3.45,4.7) {};
			\node[draw,fill=red,inner sep=2pt,circle] at (4.25,4.4) {};
			
			\node[draw,fill=red,inner sep=2pt,circle] at (5,4.4) {};
			\node[draw,fill=red,inner sep=2pt,circle] at (5.5,4.4) {};
			\node[draw,fill=red,inner sep=2pt,circle] at (6.5,4.4) {};
			\node[draw,fill=red,inner sep=2pt,circle] at (7,4.9) {};
			
			\node at (6.5,8) {(2)};
			
			\path[draw] (7.5,8.2) -> (7.5,4.1);
		\end{tikzpicture}
		\begin{tikzpicture}[scale=0.8,every node/.style={scale=0.7}]
			\draw[-{Stealth[length=6pt]}] (2,5.8) to (2,4.8);
			\draw[-{Stealth[length=6pt]}] (2.9,5.8) to node[left] {3} (3.45,4.8);
			\draw[-{Stealth[length=6pt]}] (3.5,6.3) to node[left] {2} (4.25,5.3);
			\draw[-{Stealth[length=6pt]}] (3.75,7.5) to node[left] {5} (4.25,6.5);
			\draw[-{Stealth[length=6pt]}] (4.75,7.5) to node[right] {4} (4.25,6.5);
			\draw[-{Stealth[length=6pt]}] (5.3,6.3) to (5.3,5.3);
			\draw[-{Stealth[length=6pt]}] (6.5,6.3) to (6.5,5.3);
			\draw[-{Stealth[length=6pt]}] (7,7) to node[left] {6} (7,6);
			
			\draw[-{Stealth[length=6pt]}] (4.25,6.5) -> (4.25,5.3);
			\draw[-{Stealth[length=6pt]}] (4.25,5.3) -> (4.25,4.5);
			
			\draw[-{Stealth[length=6pt]}] (7,6) -> (6.5,5.3);
			\draw[-{Stealth[length=6pt]}] (7,6) -> (7,5);
			
			\draw[-{Stealth[length=6pt]}] (6.5,5.3) -> (6.5,4.5);
			\draw[-{Stealth[length=6pt]}] (5.3,5.3) -> (5,4.5);
			\draw[-{Stealth[length=6pt]}] (5.3,5.3) -> (5.5,4.5);
			
			\node[draw,fill=red,inner sep=2pt,circle] at (2,4.7) {};
			\node[draw,fill=red,inner sep=2pt,circle] at (3.45,4.7) {};
			\node[draw,fill=red,inner sep=2pt,circle] at (4.25,4.4) {};
			
			\node[draw,fill=red,inner sep=2pt,circle] at (5,4.4) {};
			\node[draw,fill=red,inner sep=2pt,circle] at (5.5,4.4) {};
			\node[draw,fill=red,inner sep=2pt,circle] at (6.5,4.4) {};
			\node[draw,fill=red,inner sep=2pt,circle] at (7,4.9) {};
			
			\node at (4.5,6) {$C_1$};
			\node at (5.55,6) {$C_2$};
			
			\node at (6.5,8) {(3)};
		\end{tikzpicture}
		\caption{In this figure, an example for the transformation in the proof of \Cref{lem:Dbar} is given.
			A hypothetical network \Net with a coloring is given in (1).
			In (2), the network after deleting edges with green offspring is depicted, and,  
			in (3), the graph $G'$ is depicted.
			Some edges are labeled with their weights.
			Unlabeled edges have weight~1.
			The connected components $C_1$ and $C_2$ have weight~$13$ and~$3$ and value~$1$ and~$2$, respectively.}
		\label{fig:transformation}
	\end{figure}%

	Compute the set of weakly connected components of $G'$.
	For every weakly connected components $C = (V_C,E_C)$ of $G'$, proceed as follows.
	Define an item $I_C$ with \emph{weight}~$\w(E_C)$ and \emph{value}~$|Y_C|$, where $Y_C$ is the set of taxa in $V_C$.

	Let~$M$ be the set of these items.
	Now, return \yes if there is a subset of items in $M$ whose total weight is at most $\Dbar$ and whose total value is at least $\kbar=|X|-k$, and \no otherwise.
	Here, \kbar and \Dbar are from the original network~$\Net$.	
	Observe that this can be determined by solving an instance of \KP with set of items~$M$, budget $\Dbar$, and target value $\kbar$, which can be done in~$\Oh(\Dbar \cdot |N| \cdot \log(\kbar)) \in \Oh(\Dbar \cdot |X| \cdot \log(\kbar))$ time~\cite{rehs,weingartner}.\footnote{The $\log(\kbar)$-factor of the running time comes from adding $\log(\kbar)$-digit numbers and is not mentioned in the original papers~\cite{rehs,weingartner}.}

	\proofparagraph{Correctness}
	We show first that if $\Instance$ is a \yes-instance of \cMAPPD, then the algorithm return~\yes and secondly we show the converse.
	
	Assume that \Instance is a \yes-instance of \cMAPPD.
	Thus, there is a set~$S \subseteq X$ of size at most~$k$ with $\PD(S) \ge D$ and $X\setminus S$ is color-fitting.
	%
	We claim that $X\setminus S$ is also color-fitting in $G'$.
	Indeed, suppose for a contradiction that this is not the case.
	Then, there exists some vertex $v$ in $G'$ for which $v$ has an offspring in $X\setminus S$ and an offspring in $S$.
	Let~$v$ be a lowest such vertex.
	Then, $v$ does not have in-degree $0$ in $G'$, as in this case $v$ has a unique child with the same offspring.
	Furthermore, $v$ does not have any green offspring in~$\Net$, as the incoming edges of $v$ were not deleted in the construction of $G'$.
	Then, in fact $v$ has the same offspring in $\Net$ as in $G'$, implying that $X\setminus S$ is not color-fitting.
	This is a contradiction and we conclude that~$X\setminus S$ is color-fitting in~$G'$.
	
	Since $X \setminus S$ is color-fitting in $G'$, and $G'$ has no green taxa, every vertex $v$ in $G$ satisfies $\off(v) \subseteq X\setminus S$ or $\off(v) \subseteq S$.
	Each edge $e$ in $\Net$ that is not affected by $S$ is strictly affected by $X\setminus S$.
	Because $X\setminus S$ is color-fitting, each taxon in $X\setminus S$ is red and so
	$e$ is still an edge in~$G'$.
	Let $C_e$ be the connected component of $G'$ that contains $e$.
	Observe that every edge in $C_e$ is also strictly affected by $X \setminus S$---indeed, for any edge $uv$ in $C$, $u$ has all its offspring in $X\setminus S$ if and only if $v$ has all its offspring in $X\setminus S$.
	Thus, $C_e$ satisfies the conditions to be in $M$ for each edge $e$ that is not affected by $S$.
	Let $C_1,\dots,C_t$ be the unique connected components that contain the edges that are not affected by $S$.
	We conclude that $\w(C_1 \cup \dots \cup C_t) \le \Dbar$ and~$C_1 \cup \dots \cup C_t$ is the set of taxa $X\setminus S$, which has size at least $\kbar$.
	Hence, $I_{C_1},\dots,I_{C_t}$ is a solution for \KP on $(M, \Dbar, \kbar)$ and the algorithm returns \yes.
	
	For the converse, assume that the algorithm returns \yes and let $I_{C_1},\dots,I_{C_t}$ be a solution for \KP on $(M, \Dbar, \kbar)$.
	Let $Y_i$ be the set of taxa of $C_i$ and define $Y := \bigcup_{i=1}^t Y_i$.
	We prove that $S := X \setminus Y$ is a solution for \cMAPPD on instance \Instance. 
	By the construction, we conclude that all taxa in~$Y_i$ are colored red and that $Y_i$ and $Y_j$ are disjoint for any $i\neq j$.
	Each vertex $v \in \Net$ that has some but not all offspring in $Y_i$ has at least one green offspring.
	Consequently, $Y$ is color-fitting.
	Further, for any edge $e = (u,v)$ that is strictly affected by $Y$ in $\Net$, we have that $v$  has no green offspring, and therefore $e$ is not deleted in the construction of $G'$.
	Moreover, as $v$ has offspring in $Y_i$ for some $i \in [t]$, we conclude that $e$ is in $C_i$.
	Because $\sum_{i=1}^t \w(E_{C_i}) \le \Dbar$ the phylogenetic diversity of $S$ is $\PD(S) = \PD(X) - \sum_{i=1}^t \w(E_{C_i}) \ge \PD(X) - \Dbar = D$.
	Likewise, as $\sum_{i=1}^t |Y_i| \ge \kbar$, we conclude $|S| = |X| - \sum_{i=1}^t |Y_i| \le |X| - \kbar = k$.

	\proofparagraph{Running Time}
	The graph $G'$ can be computed from $\Net$ in~$\Oh(m)$~time.
	The weakly connected components of $G'$ can be computed in~$\Oh(m)$~time as well.
	For each connected component~$C = (V_C,E_C)$ the item~$I_C$ of~$M$
	is computed in~$\Oh(|E_C| \cdot \log(\Dbar))$~time.
	Consequently, we can compute~$M$ in~$\Oh(m \cdot \log(\Dbar))$~time.
	As the instance of \KP can be solved in~$\Oh(\Dbar \cdot |X| \cdot \log(\kbar))$~time~\cite{weingartner,rehs}, we have an overall running time of $\Oh(m \cdot \log(\Dbar) + \Dbar \cdot |X| \cdot \log(\kbar)) \in \Oh(\Dbar \cdot m \cdot \log(\Dbar))$.
\end{proof}

To show that \MAPPD is \FPT with respect to \Dbar, we present an reduction from \MAPPD to \cMAPPD using standard color coding techniques.
In particular, we show that there exists a family ${\cal F}$ of $2$-colorings  $c:E\to \{\text{red},\text{green}\}$, with $|{\cal F}|$ bounded by a function of $\Dbar$ times a polynomial in~$n$, such that $(\Net, k, D)$ is a \yes-instance of \MAPPD if and only if $(\Net,k,D,c)$ is a \yes-instance of \cMAPPD for some $c \in {\cal F}$.
We refer the reader to~\cite[Chapter 5.2]{cygan} for an overview of color coding.
 
\begin{definition}
	An \emph{$(n,k)$-universal set} is a family ${\cal U}$ of subsets of $[n]$ such that for any~$S\subseteq [n]$ of size $k$, $\{A \cap S \mid A \in {\cal U}\}$ contains all $2^k$ subsets of $S$.
\end{definition}
 
\begin{theorem}[\cite{Naor1995SplittersAN}]
	\label{thm:universalSet}
	For any $n,k \geq 1$, one can construct an $(n,k)$-universal set of size~$2^k k^{\Oh(\log k)}\log n$ in time $2^kk^{\Oh(\log k)}n\log n$.
\end{theorem}

\begin{proof}[Proof of \Cref{thm:Dbar}]
	\proofparagraph{Algorithm}
	Let $\Instance := (\Net:=(V,E,\w),k,D)$ be an instance of \MAPPD.
	Arbitrarily order the taxa $x_1,\dots, x_n$. 
	Construct an~$(n,\Dbar + \kbar)$-universal
	
	Now for each $A \in \mathcal{U}$, construct a $2$-coloring  $c_A: X\to \{\text{red},\text{green}\}$ where $x_i$ is colored green if and only if $i \in A$, and solve \cMAPPD on $(\Net, k, D, c_A)$.
	Return \yes, if~$(\Net, k, D, c_A)$ is a \yes-instance for some~$A \in \mathcal{U}$. Otherwise return \no.
	
	\proofparagraph{Correctness}
	First observe that if $(\Net, k, D, c)$ is a \yes-instance of \cMAPPD for any coloring $c: X\to \{\text{red},\text{green}\}$, then $(\Net,k,D)$ is also a \yes-instance of \MAPPD.
	
	Now suppose $(\Net,k,D)$ is a \yes-instance for \MAPPD.
	Let $S\subseteq X$ be a subset of taxa of size at most~$k$ and with~$\PD(S)\ge D$.
	If necessary, add taxa to~$S$ until~$|S| = k$.
	Consequently, $X\setminus S$ has size $\kbar$.
	Let $V_Y$ be the set of vertices~$u$ of \Net which have an offspring~$x_u$ in~$S$ and have a child~$v$ with~$\off(v) \subseteq X\setminus S$.
	Define~$Y := \{ x_u \mid u\in V_Y \}$.
	Observe that if we could define a coloring which colors the taxa in~$X\setminus S$ in red and the taxa in~$Y$ in green, then~$X\setminus S$ would be color-fitting.
	
	Define an operation $\Index : 2^X \to [n]$ by~$\Index(X') := \{ i \mid x_i \in X' \}$ for any set~$X' \subseteq X$.
	
	For each~$u\in V_Y$ and child~$v$ of~$u$ with~$\off(v) \subseteq X\setminus S$,
	the edge~$uv$ is strictly affected by~$X\setminus S$.
	Consequently, $|Y| \le |V_Y| \le \w(uv) \le \Dbar$.
	Therefore, $Z := Y \cup (X\setminus S) \subseteq X$ is a set of size at most~$\kbar + \Dbar$.
	If necessary, add taxa until~$Z$ has size~$\Dbar+\kbar$.
	Consequently, there is a set~$A\in \mathcal{U}$ with~$A \cap \Index(Z) = \Index(Y)$.
	Then, the associated $2$-coloring $c_A$ colors every taxa in $Y$ green and every taxa in $X \setminus S$ red.
	So $X\setminus S$ is color-fitting with respect to $c_A$, and $S$ is a solution for \cMAPPD on instance~$(\Net, k, D, c_A)$.
	
	\proofparagraph{Running Time}
	The construction of $\mathcal{U}$ takes $2^{\Dbar + \kbar + \Oh(\log^2 (\Dbar))} n\log n$ time, and for each of the $2^{\Dbar + \kbar + \Oh(\log^2 (\Dbar))}\log n$ sets in $\mathcal{U}$ we solve an instance of \cMAPPD, which can be done in $\Oh(\Dbar \cdot m \cdot \log(\Dbar))$~time by~\Cref{lem:Dbar}.
	The overall running time is~$\Oh(2^{\Dbar + \kbar + o(\Dbar)} \cdot m\log n)$.
\end{proof}

\subsection{Proximity to trees}
\label{sec:FPTtreelike}
\MAPPD can be solved in polynomial time with Faith's Greedy-Algorithm, if the given network is a tree \cite{FAITH1992,Pardi2005,steel}.
Therefore, in this subsection, we examine \MAPPD with respect to two parameters that classify the network's proximity to a tree, the number of reticulations $\vret$ and the smaller parameter treewidth $\tw$.

\begin{theorem}
	\label{thm:ret}
	\MAPPD can be solved in $\Oh(2^{\vret} \cdot k \cdot m\cdot \log(\max_\w))$ time.
\end{theorem}
Observe that by \Cref{cor:X}, \MAPPD can not be solved in $\Oh^*(2^{\epsilon\cdot \vret})$ time for any $\epsilon<1$, unless \SETH fails.
Therefore, the running time of the previous proof is tight, to some extent.

We note that after the initial publication of this paper, it was shown that \MAPPD can be solved in $\Oh(2^{\text{sw}_\Net})$ time, where~$\text{sw}_\Net$ is the scanwidth of the network~\cite{PaNDA}.
A network's scanwidth is at least its treewidth and at most its number of reticulations minus~1.
Consequently, the algorithms in~\cite{PaNDA} and~\cite{holtgrefe2026tractable} present neat improvement to our algorithm presented in~\Cref{thm:ret} and demonstrates the algorithmic power behind scanwidth.
For more detains and a formal definition of scanwidth, we refer to~\cite{berry,holtgrefeThesis}.
\begin{proof}
	\proofparagraph{Algorithm}
	We define two operations, called \take and \leave, that take as input an instance $\Instance = (\Net,k,D)$ and a reticulation $v$ of \Net and return another instance of \MAPPD.
	Every subset of taxa $Y$ that contains an offspring of $v$ should be a solution for \MAPPD on~\Instance if and only if $Y$ is a solution for\MAPPD on $\take(\Instance,v)$.
	Similarly every subset of $Y$ that does not contain an offspring of $v$ should be a solution for \MAPPD on \Instance if and only if $Y$ is a solution for\MAPPD on $\leave(\Instance,v)$.

	Recall that $\off(e)\subseteq X$ is the set of offspring of $w$ for an edge $e=(v,w)$ and
	the strictly affected edges $T_Y$ for a set of taxa $Y\subseteq X$ is the set of edges $e$ with $\off(e)\subseteq Y$.
	We define $\leave(\Instance,v)$ to be the instance $\Instance'=(\Net',k,D)$ of \MAPPD, in which $k$ and $D$ are unchanged and $\Net'$ is the  network that results from deleting the edges $T_{\off(v)}$ and the resulting isolated vertices from \Net.

	For a reticulation $v$ in a network \Net with child $u$, let $E^{(\uparrow vu)}$ be the set of edges of~\Net that are between two vertices of $\anc(v) \cup \{u\}$.
	Recall that $\Dbar=\sum_{e\in E} \w(e)-D$.
	We define $\take(\Instance,v)$ to be the instance $\Instance'=(\Net',k,D')$ of \MAPPD with $D'=D+\Dbar$ and $k$ is unchanged.
	$\Net'$ is the network that results from \Net by deleting the edges $E^{(\uparrow vu)}$, merging all the ancestors of $v$ to a single vertex $\rho$, adding an edge $(\rho,u)$, and setting the weight $\w'((\rho,u))$ to $\w(E^{(\uparrow vu)})+\Dbar$.
	For each vertex $w\ne u$ that, in~\Net, has $t\ge 1$ parents $u_1,\dots,u_t$ in $\anc(v)$, we replace the edges $(u_1,w),\dots, (u_t,w)$ with a single edge $(\rho,w)$ that has weight $\sum_{i=1}^t \w((u_i,w))$.
	Observe that $PD_{\Net'}(X) = \PD(X) + \Dbar$.
	\Cref{fig:leave+take} depicts an example of the operations \take and \leave.
	\begin{figure}[t]
		\centering
		\begin{tikzpicture}[scale=0.8,every node/.style={scale=0.7}]
			\draw[-{Stealth[length=6pt]}] (5,8) -> (7,7);
			\draw[-{Stealth[length=6pt]}] (5,8) -> node[left] {3} (4.5,7);
			\draw[-{Stealth[length=6pt]}] (5,8) -> (3,7);
			\draw[-{Stealth[length=6pt]}] (5,8) -> node[right] {2} (6,6);
			
			\draw[-{Stealth[length=6pt]}] (3,7) -> node[left] {2} (2,6);
			\draw[-{Stealth[length=6pt]}] (3,7) -> node[left] {3} (3.75,6);
			\draw[-{Stealth[length=6pt]}] (4.5,7) -> node[right] {4} (4.5,6);
			\draw[-{Stealth[length=6pt]}] (4.5,7) -> node[above] {2} (6,6);
			\draw[-{Stealth[length=6pt]}] (4.5,7) -> (3.75,6);
			\draw[-{Stealth[length=6pt]}] (7,7) -> (6,6);
			\draw[-{Stealth[length=6pt]}] (7,7) -> (7,6);
			
			\draw[-{Stealth[length=6pt]}] (2,6) -> (2,5);
			\draw[-{Stealth[length=6pt]}] (2,6) -> (3,5);
			\draw[-{Stealth[length=6pt]}] (2,6) -> node[above] {2} (3.75,5.4);
			\draw[-{Stealth[length=6pt]}] (3.75,6) -> (3.75,5.4);
			\draw[-{Stealth[length=6pt]}] (4.5,6) -> (4.5,5);
			\draw[-{Stealth[length=6pt]}] (4.5,6) -> (6,5.4);
			\draw[-{Stealth[length=6pt]}] (6,6) -> (6,5.4);
			
			\draw[-{Stealth[length=6pt]}] (3.75,5.4) -> (3.75,4.8);
			\draw[-{Stealth[length=6pt]}] (6,5.4) -> (6,4.8);
			
			\node at (6.5,8) {(1)};
			\node at (4,6) {$v$};
			\node at (4,5.4) {$u$};
			\node at (6.3,5.9) {$w$};
			\node at (2.5,8) {$k=3$};
			\node at (2.5,7.5) {$D=28$};
			
			\path[draw] (7.5,8) -> (7.5,4.8);
		\end{tikzpicture}
		\begin{tikzpicture}[scale=0.8,every node/.style={scale=0.7}]
			\draw[-{Stealth[length=6pt]}] (5,8) -> (7,7);
			\draw[-{Stealth[length=6pt]}] (5,8) -> node[left] {3} (4.5,7);
			\draw[-{Stealth[length=6pt]}] (5,8) -> (3,7);
			\draw[-{Stealth[length=6pt]}] (5,8) -> node[right] {2} (6,6);
			
			\draw[-{Stealth[length=6pt]}] (3,7) -> node[left] {2} (2,6);
			\draw[-{Stealth[length=6pt]}] (4.5,7) -> node[left] {4} (4.5,6);
			\draw[-{Stealth[length=6pt]}] (4.5,7) -> node[above] {2} (6,6);
			\draw[-{Stealth[length=6pt]}] (7,7) -> (6,6);
			\draw[-{Stealth[length=6pt]}] (7,7) -> (7,6);
			
			\draw[-{Stealth[length=6pt]}] (2,6) -> (2,5);
			\draw[-{Stealth[length=6pt]}] (2,6) -> (3,5);
			\draw[-{Stealth[length=6pt]}] (4.5,6) -> (4.5,5);
			\draw[-{Stealth[length=6pt]}] (4.5,6) -> (6,5.4);
			\draw[-{Stealth[length=6pt]}] (6,6) -> (6,5.4);
			
			\draw[-{Stealth[length=6pt]}] (6,5.4) -> (6,4.8);
			
			\node at (6.5,8) {(2)};
			\node at (2.5,8) {$k=3$};
			\node at (6.3,5.9) {$w$};
			\node at (2.5,7.5) {$D=28$};
			
			\path[draw] (7.5,8) -> (7.5,4.8);
		\end{tikzpicture}
		\begin{tikzpicture}[scale=0.8,every node/.style={scale=0.7}]
			\draw[-{Stealth[length=6pt]}] (5,8) -> (7,7);
			\draw[-{Stealth[length=6pt]}] (5,8) -> node[right] {4} (6,6);
			\draw[-{Stealth[length=6pt]}] (5,8) -> node[right] {2} (2,6);
			\draw[-{Stealth[length=6pt]}] (5,8) -> node[right] {4} (4.5,6);
			\draw[-{Stealth[length=6pt]}] (5,8) -> node[left] {12} (3.75,5.4);
	
			\draw[-{Stealth[length=6pt]}] (7,7) -> (6,6);
			\draw[-{Stealth[length=6pt]}] (7,7) -> (7,6);
			
			\draw[-{Stealth[length=6pt]}] (2,6) -> (2,5);
			\draw[-{Stealth[length=6pt]}] (2,6) -> (3,5);
			\draw[-{Stealth[length=6pt]}] (2,6) -> node[above] {2} (3.75,5.4);
			\draw[-{Stealth[length=6pt]}] (4.5,6) -> (4.5,5);
			\draw[-{Stealth[length=6pt]}] (4.5,6) -> (6,5.4);
			\draw[-{Stealth[length=6pt]}] (6,6) -> (6,5.4);
			
			\draw[-{Stealth[length=6pt]}] (3.75,5.4) -> (3.75,4.8);
			\draw[-{Stealth[length=6pt]}] (6,5.4) -> (6,4.8);
			
			\node at (6.5,8) {(3)};
			\node at (4,5.4) {$u$};
			\node at (6.3,5.9) {$w$};
			\node at (2.5,8) {$k=3$};
			\node at (2.5,7.5) {$D=31$};
		\end{tikzpicture}
		\caption{In this figure, an example for the usage of $\leave$ and $\take$ is given.
			A hypothetical instance \Instance is given in (1).
			Here, the value of $\Dbar$ is $3$.
			In (2) the instance $\leave(\Instance,v)$, and in (3) the instance $\take(\Instance,v)$ is depicted.
			Unlabeled edges have a weight of 1.
			Observe in (3), the weight of the edge $(\rho,w)$ is 4, as $w$ has two edges from ancestors of $v$ in \Instance which have a weight of 2 each.
			The weight of $(\rho,u)$ is $12$, as in \Instance the edges of $E^{(\uparrow vu)}$ have a combined weight of 9.}
		\label{fig:leave+take}
	\end{figure}
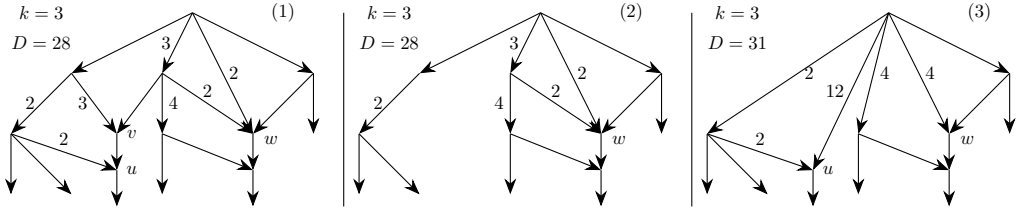%

	Now, we are at the position to define the branching algorithm. Let $\Instance=(\Net,k,D)$ be an instance of \MAPPD.
	If $\Net$ is a phylogenetic tree, solve the instance \Instance with Faith's Algorithm~\cite{FAITH1992,Pardi2005,steel}.
	Otherwise, let $v$ be a reticulation of $\Net$.
	
	Then, recursively, solve \MAPPD on $\take(\Instance,v)$ and $\leave(\Instance,v)$.
	If either of these is a \yes-instance of \MAPPD, return \yes.
	Otherwise, return \no.

	\proofparagraph{Correctness}
	The correctness of the base case is given by the correctness of Faith's Algorithm.
	We show that if \Net contains a reticulation $v$, then \Instance is a \yes-instance of \MAPPD if and only if $\take(\Instance,v)$ or $\leave(\Instance,v)$ is a \yes-instance of \MAPPD.

	Consider any set of taxa $Y \subseteq X$.
	Firstly, we claim that if $Y \cap \off(v) = \emptyset$, then $PD_{\Net'}(Y) = \PD(Y)$, where $\Net'$ is the network in $\leave(\Instance,v)$.
	Indeed, $\Net'$ contains all the vertices and edges of \Net that have an offspring outside of $\off(v)$.
	Therefore, the set of edges affected by $Y$ is the same in both networks, and $PD_{\Net'}(Y) = \PD(Y)$.
	Secondly, we claim that if  $Y \cap \off(v) \neq \emptyset$, then $PD_{\Net'}(Y) = \PD(Y) + \Dbar$, where $\Net'$ is the network in $\take(\Instance,v)$.
	Recall that each edge $e=(u_1,u_2)$ with $u_1\ne \rho$ of $E(\Net')$ is also an edge of $\Net$ and $\w'(e)=\w(e)$.
	Further, for each edge $e=(\rho,u_2)\in E(\Net')$ with $u_2\ne u$, there are edges $e_1=(u_{i_1},u_2),\dots,e_t=(u_{i_t},u_2)$ of $E(\Net)$ with $\w'(e)=\sum_{i=1}^t \w(e_i)$.
	Now, let $Q = Q_1 \cup Q_2 \cup \{(\rho,u)\}$ be the edges of $\Net'$ that have at least one offspring in $Y$, of which edges in $Q_1$ have both endpoints in $V(\Net')\setminus\{\rho\}$, and $Q_2$ are outgoing edges of~$\rho$.
	Further, let $P = P_1 \cup P_2 \cup E^{(\uparrow vu)}$ be the edges of $\Net$ that have at least one offspring in~$Y$, of which edges in $P_1$ have both endpoints in $V(\Net')\setminus \{\rho\}$, and $P_2$ are edges with one endpoint in $\anc(v)\setminus \{v\}$ and one endpoint in $V(\Net')\setminus\{\rho\}$.
	Observe that, $Q_1 = P_1$ and $\w'(Q_1) = \w(P_1)$ since any vertex in $V(\Net')\setminus \{\rho\}$ has the same offspring in $\Net$ as in $\Net'$. 
	Further, $\w'(Q_2) = \w(P_2)$ as for each $u_2 \in V(\Net')\setminus\{\rho\}$, the total weight of edges $(u_1,u_2)$ with $u_1 \in \anc(v)\setminus \{v\}$ in $\Net$ is equal to the weight of the edge $(\rho,u_2)$ in $\Net'$.
	It follows that $PD_{\Net'}(Y) = \w'(Q_1) + \w'(Q_2) + \w'((\rho,u)) = \w(P_1) + \w(P_2) + \w(E^{(\uparrow vu)}) + \Dbar = \PD(Y) + \Dbar$.
	%
	%
	
	It follows from the above that if $Y$ is a solution for \MAPPD on \Instance (that is, $|Y|\le k$ and $\PD(Y)\ge D$), then either $Y$ is a solution for \MAPPD on $\leave(\Instance,v)$ or $Y$ is a solution for \MAPPD on $\take(\Instance,v)$.
	Conversely, if~$Y$ is a solution for \MAPPD on~$\leave(\Instance,v)$ then $Y \cap \off(e) = \emptyset$ and thus $ \PD(Y) = PD_{\Net'}(Y) \geq D$, so~$Y$ is also a solution for \MAPPD on~\Instance.
	Finally, if $Y$ is a solution for \MAPPD on $\take(\Instance,v)$ then $Y \cap \off(e) \neq \emptyset$, as otherwise $PD_{\Net'}(Y) \le PD_{\Net'}(X) - \w'((\rho,v)) =  D + 2\Dbar  - (\w(E^{(\uparrow vu)}) + \Dbar) \le D + \Dbar - 1 < D'$.
	Then $PD_{\Net'}(Y) = \PD(Y) + \Dbar$, from which it follows that $\PD(Y) \geq D' - \Dbar = D$ and $Y$ is also a solution for \MAPPD on \Instance.

	%
	%
	%
	%
	%
	%
	%

	\proofparagraph{Running Time}
	Let \Instance be an instance of \MAPPD that contains a reticulation $v$.
	The number of reticulations in \Instance is greater than the number of reticulations in $\take(\Instance,v)$ and $\leave(\Instance,v)$, because at least the reticulation, $v$, is removed and no new reticulations are added.
	Therefore, the search tree contains $\Oh(2^{\vret})$ nodes.
	It can be checked in $\Oh(m)$ time, if $\Net$ contains a reticulation.
	Faith's Algorithm takes $\Oh(k \cdot m\cdot \log(\max_\w))$ time \cite{Pardi2005,steel}.
	
	The sets $\off(v)$ and $\anc(v)$ for a vertex $v$, and $T_Y$ for a set $Y$ can be computed in $\Oh(m)$ time.
	Once $\anc(v)$ is computed, we can iterate over $E$ to find the edges that are outgoing from $\anc(v)$ and compute the total weight for an edge $(\rho,w)$ in $\Net'$ in $\Oh(m\cdot \log(\max_\w))$ time, which is also the time needed to compute $\w((\rho,u))$ which needs \Dbar and the weight of $E^{(\uparrow vu)}$.
	Therefore, the instances $\take(\Instance,v)$ and $\leave(\Instance,v)$ can be computed in $\Oh(m\cdot \log(\max_\w))$ time.
	
	Thus, a solution for \MAPPD can be computed in $\Oh(2^{\vret} \cdot k \cdot m\cdot \log(\max_\w))$ time.
\end{proof}

Bordewich et al. showed that \MAPPD can be solved in polynomial time on level-1 networks \cite{bordewich}.
We extend this result by showing that \MAPPD is fixed-parameter tractable with respect to treewidth.
We note that for every network, the number of reticulations is at least the level, and the level plus one is at least the treewidth.

\begin{theorem}
	\label{thm:tw}
	\MAPPD can be solved in $\Oh(9^\tw \cdot \tw \cdot k^2 \cdot m  \cdot \log D)$ time.
\end{theorem}
We first provide a sketch of the main ideas here.

We aim to find a set of edges $E'$ that have an overall weight of at least $D$ and that are incident with at most $k$ leaves.
Further, for each edge $e=(u,v)\in E'$ we require that either~$v\in X$ or there is an edge $(v,w)\in E'$.
In the algorithm, which is a dynamic program over a nice tree decomposition, we index feasible partial solutions by a 3-coloring of the vertices.
At a given node of the tree decomposition, a vertex $v$ is colored:
\begin{itemize}
	\item red, if it is still mandatory that we select an outgoing edge of $v$ (because we have selected an incoming edge of $v$),
	\item green, if we can select incoming edges of~$v$ and do not need to select an outgoing edge of~$v$ (because $v$ is a leaf or we have already selected an outgoing edge of $v$),
	\item black, if we have to not yet selected an edge incident with~$v$ (such that only the selection of an incoming edge of~$v$ makes the selection of an outgoing edge of $v$ necessary).
\end{itemize}
We introduce each leaf as a green vertex and the other vertices as black vertices.
In order to consider only feasible solutions, a vertex must be green or black when it is forgotten.
The most important step of the algorithm is in the introduction of an edge, where colors may be adjusted depending on whether or not the new edge is included in $E'$.

For the proof, we require the notion of a \emph{nice tree decomposition}.
\begin{definition}[Nice tree decomposition, treewidth]
	\label{def:tw}
	A tree decomposition of a graph~$G$ consists of a rooted tree~$T$ and, for every node $t \in T$ a bag $Q_t \subseteq V(G)$
	such that every edge of~$G$ is contained in some bag and such that all the nodes of~$T$ whose bags contain any particular vertex~$v$ is connected.
	
	A tree decomposition is called \emph{nice} \cite{CyganTreeDecomp} if
	the bags of the root and all leafs of~$T$ are empty,
	and every non-leaf bag $Q_t$ has one of the following types
	\begin{description}
		\item[introduce vertex bag:] $t$ has only one child~$t'$ and $Q_{t'} = Q_t \setminus \{v\}$ for some vertex~$v$ which is said to \emph{be introduced} at~$t$;
		\item[introduce edge bag:] $t$ has only one child~$t'$ and $Q_{t'} = Q_t$ and $t$ is labeled with an edge $e \subseteq Q_t$ which is said to \emph{be introduced} at~$t$.
		\item[forget bag:] $t$ has only one child $t'$ and $Q_{t'} = Q_t \cup \{v\}$ for some vertex~$v$;
		\item[join bag:] $t$ has exactly two children $t_1, t_2$ and $Q_t = Q_{t_1} = Q_{t_2}$;
	\end{description}
	Every edge of~$G$ must be introduced exactly once, and we may assume this happens as high in~$T$ as possible,
	i.e., just before some endpoint of that edge is forgotten.
	For a node $t\in T$, we denote by $G_t = (V_t, E_t)$ the subgraph
	of~$G$ that contains exactly those vertices and edges which are introduced at~$t$ or any descendant of~$t$.
	Note that at every join bag the graph $G_t[Q_t]$ is edgeless.
\end{definition}
\begin{proof}[Proof of \Cref{thm:tw}]
	We define a dynamic programming instance over a nice tree-de\-com\-po\-si\-tion~$T$ of the $X$-network $\Net=(V,E,\w)$ in which vertices are introduced without incident edges and all edges are induced exactly once.
	
	\proofparagraph{Definition of the Table}
	Let $\Instance=(\Net,k,D)$ be an instance of \MAPPD and let~$T$ be a nice tree decomposition of  $\Net$.
	We index solutions by a node $t \in T$, a partition $R\cup G\cup B$ of $Q_t$, and a non-negative integer $s\in [k]_0$.
	For a set of edges~$F \subseteq E_t$, we call a vertex $u \in V_t$ \emph{\emph{green} with respect to $F$} if $u$ is a leaf or has an outgoing edge in $F$. We call $u$ \emph{\emph{red} with respect to $F$} if $u$ is not a leaf and has an incoming but no outgoing edge in $F$. Finally, we call $u$  \emph{\emph{black} with respect to $F$} if $u$ is not a leaf and has no incident edges in $F$.
	For a node~$t\in T$, a partition $R\cup G\cup B$ of $Q_t$ and an integer $s$, we call a set of edges~$F \subseteq E_t$ \emph{feasible for~$t,R,G,B,s$}, if all the following conditions hold:
	\begin{enumerate}
		\item[(F1)]\label{it:safeness} If $(u,v)$ is an edge in $F$ and $v\notin Q_t$, then $v$ is a leaf of $\Net$ or $v$ has an outgoing edge in $F$.
		\item[(F2)]\label{it:red} The vertices $R\subseteq Q_t$ are red with respect to $F$.
		\item[(F3)]\label{it:green} The vertices $G\subseteq Q_t$ are green with respect to $F$. 
		\item[(F4)]\label{it:black} The vertices $B\subseteq Q_t$ are black with respect to $F$.
		\item[(F5)]\label{it:taxa} The number of leaves in $\Net$ with an incoming edge in $F$ is $s$.
	\end{enumerate}
	Further, we define \Sstar to be the set of all sets~$F$ that are feasible for $t,R,G,B,s$.
	%

	We define a dynamic programming algorithm over a nice tree decomposition~$T$.
	In a table entry of $\DP[t,A,R,G,B,s]$, we store the greatest weight $\w(F)$ of a set $F \subseteq E_t$ that is feasible for~$t,R,G,B,s$. If there is no feasible $F$, then we store a large negative value.
	For our purposes a value of $-m\cdot \max_\w-1$ will suffice, as even if every edge would be chosen, the entry at the root of the tree decomposition is still negative.
	Let~$r$ be the root of the nice tree-decomposition $T$.
	We claim that $\DP[r,\emptyset,\emptyset,\emptyset,k]$ stores the greatest phylogenetic diversity that can be preserved with a budget of $k$.
	Indeed, for any set $Y$ of $k$ taxa, the set of affected edges $E_Y$ is feasible for $r,\emptyset,\emptyset,\emptyset,k$ and has weight $\PD(Y)$. Conversely if $F$ is feasible for~$r,\emptyset,\emptyset,\emptyset,k$, then letting $Y$ be the set of leaves in $\Net$ with an incoming edge in $F$, we have that $E_Y \supseteq F$ and so $\PD(Y)\geq \w(F)$.
	Hence, we can return \yes if $\DP[r,\emptyset,\emptyset,\emptyset,k]\ge D$ and \no otherwise.

	Now we have everything we need to define the dynamic programming algorithm.
	In the calculations that follow, any time a value $\DP[t,R,G,B,s]$ is called for which $\DP[t,R,G,B,s]$ is not defined (in particular, if $s < 0$), we take $\DP[t,R,G,B,s]$ to be $-m\cdot \max_\w-1$.

	\proofparagraph{Leaf Node}
	For a leaf~$t$ of~$T$ the sets~$Q_t$ and $E_t$ are empty.
	So if $s = 0$, we trivially store
	\begin{eqnarray} \label{tw:leaf}
	\DP[t,\emptyset,\emptyset,\emptyset,s] &=& 0.
	\end{eqnarray}
	Otherwise, we store $\DP[t,R,G,B,s] = -m\cdot \max_\w-1$.

	\proofparagraph{Introduce Vertex Node}
	Suppose now that~$t$ is an \emph{introduce vertex node}, i.e.,~$t$ has a single child~$t'$,~$Q_t = Q_{t'} \cup \{v\}$, and $v$ is an isolated vertex in $\Net_t$.
	If either $v\in B$, or $v\in G$ and $v$ is a leaf, we store
	\begin{eqnarray} \label{tw:insertvertex}
	\DP[t,R,G,B,s] &=& \DP[t',R,G\setminus\{v\},B\setminus\{v\},s].
	\end{eqnarray}
	Otherwise, we store $\DP[t,R,G,B,s] = -m\cdot \max_\w-1$.

	\proofparagraph{Introduce Edge Node}
	Suppose now that~$t$ is an \emph{introduce edge node}, i.e.,~$t$ has a single child~$t'$,~$Q_t = Q_{t'}$, and $e=(v,w)$ is introduced at $t'$.
	The algorithm must decide whether $e$ is an affected edge, or not.
	If $v\notin G$ or $w\in B$,  $e$ can not be an affected edge and so we store $\DP[t,R,G,B,s] = \DP[t',R,G,B,s]$.
	Otherwise, we store
	\begin{eqnarray} \label{tw:insertedge}
	\DP[t,R,G,B,s] &=& \max\{
	\DP[t',R,G,B,s];
	\max_{R',G',B'}\DP[t',R',G',B',s'] + \w(e)
	\},
	\end{eqnarray}
	where the second maximum is over all possible partitions $R'\cup G'\cup B'$ of $Q_t$, such that $S\setminus\{v,w\} = S'\setminus\{v,w\}$ for all $S \in \{R,G,B\}$ (i.e. the two partitions agree on $Q_T\setminus\{v,w\}$), and such that if $w \in G'$ then $w \in G$, and $w \in R$ otherwise.
	Here $s' = s-1$ if $w$ is a leaf and~$s' = s$ if not.

	\proofparagraph{Forget Node}
	Suppose now that~$t$ is an \emph{forget node}, i.e.,~$t$ has a single child~$t'$ and~$Q_t = Q_{t'} \setminus \{v\}$.
	We store
	\begin{eqnarray} \label{tw:forget}
	\DP[t,R,G,B,s] =  \max \{
	\DP[t',R,G,B\cup\{v\},s];
	\DP[t',R,G\cup\{v\},B,s]\}.
	\end{eqnarray}

	\proofparagraph{Join Node}
	Suppose now that~$t\in T$ is an \emph{join node}, i.e.,~$t$ in $T$ has two children~$t_1$ and~$t_2$ with~$Q_t = Q_{t_1} = Q_{t_2}$.
	We call two partitions $R_1\cup G_1\cup B_1$ and $R_2\cup G_2 \cup B_2$ of $Q_t$ \emph{qualified} for $R\cup G\cup B$ if $R=(R_1\cup R_2) \setminus (G_1 \cup G_2)$ and $G = G_1\cup G_2$ (and consequently $B = B_1 \cap B_2$). See~\Cref{fig:joinColorings}.
	We store
	\begin{eqnarray} \label{tw:join}
	\DP[t,R,G,B,s] &=& \max_{(\pi_1, \pi_2) \in{\cal Q},s'} ~ \DP[t_1,R_1,G_1,B_1,s'] + \DP[t_2,R_2,G_2,B_2,s-s'],
	\end{eqnarray}
	where ${\cal Q}$ is the set of pairs of partitions $\pi_1 = R_1\cup G_1\cup B_1$ and $\pi_2 = R_2\cup G_2 \cup B_2$ that are qualified for $R,G,B$ and $s' \in [s]_0$.
	
	\newcommand{\darkgreen}{green!50!black}
	\begin{figure}
	\begin{center}
	\begin{tabular}{c|c|c|c|}\cline{2-4}
	& $B_1$ & \color{red}{$R_1$} & \color{\darkgreen}{$G_1$}\\\hline
	\multicolumn{1}{|l|}{$B_2$}  & $B$ & \color{red}{$R$} & \color{\darkgreen}{$G$} \\\hline
	\multicolumn{1}{|l|}{\color{red}{$R_2$}}  & \color{red}{$R$} & \color{red}{$R$} & \color{\darkgreen}{$G$} \\\hline
	\multicolumn{1}{|l|}{\color{\darkgreen}{$G_2$}}  & \color{\darkgreen}{$G$} & \color{\darkgreen}{$G$} & \color{\darkgreen}{$G$} \\\hline
	\end{tabular} 
	
		\caption{This table shows the relationship between the partitions $R_1\cup G_1\cup B_1$, $R_2\cup G_2 \cup B_2$ and $R\cup G\cup B$ in the case of a join node, when $R_1\cup G_1\cup B_1$ and $R_2\cup G_2 \cup B_2$ are qualified for $R\cup G\cup B$. The table shows which of the set $R,G$, or~$B$ an element $v \in Q_t$ will be in, depending on its membership in $R_1,G_1,B_1,R_2,G_2$, and $B_2$. For example if $v \in R_1$ and $v \in B_2$, then $v \in R$.}
		\label{fig:joinColorings}
	\end{center}
	\end{figure}

	\proofparagraph{Correctness}
	Let $t$ be a node of the nice tree-decomposition $T$, $s$ an integer, and $R\cup G\cup B$ a partition of $Q_t$.
	We show that the value of $\DP[t,R,G,B,s]$ is correct, for each type of node individually.
	As we have already argued that $\DP[r,\emptyset,\emptyset,\emptyset,k]$ stores the greatest phylogenetic diversity that can be preserved with a budget of $k$, this is sufficient to prove the correctness of the algorithm.

	If $t$ is a \emph{leaf node}, then $Q_t$ and $E_t$ are empty.
	Consequently, $R=G=B=\emptyset$ and any feasible set $F$ is also empty.
	Therefore, we also conclude with (F5) that $s=0$ and $\w(F)=0$ for any $F \in \Sstar$.
	Hence, the stored value by Equation (\ref{tw:leaf}) is correct.

	Let $t$ be an \emph{introduce vertex node} with child $t'$ and $Q_t=Q_{t'}\cup\{v\}$.
	Then, the graph~$\Net_t$ is the graph $\Net_{t'}$ with an additional isolated vertex $v$.
	Thus, $v$ is black with respect to any $F \subseteq E_t$, unless $v$ is a leaf in which case $v$ is green.
	Thus, there is no feasible set $F$ for $\Sstar$ unless either $v \in B$ or $v \in G$ and $v$ is a leaf.
	Assuming this is the case, any set $F \subseteq E_t$ is feasible for $t,R,G,B,s$ if and only if it is feasible for $t',R,G\setminus\{v\},B\setminus\{v\},s$.
	Hence, we store the right value.

	Let $t$ be a \emph{introduce edge node} with child $t'$ and $Q_t=Q_{t'}$ and $\Net_t-e=\Net_{t'}$. Define $e:=(v,w)$.
	Clearly, every set $F\subseteq E_t\setminus\{e\}$  is feasible for $t,R,G,B,s$ if and only if $F$ is also feasible for $t',R,G,B,s$, because then $\Net_t[F]$ is isomorphic to $\Net_{t'}[F]$.
	%
	Now consider a set~$F\subseteq E_t$ with $(v,w) \in F$.
	Note that $v$ is green, and $w$ cannot be black, with respect to any such set $F$. Therefore such an $F$ only exists if $v \in G$ and $w \notin B$.
	Now let $F' = F\setminus\{(v,w)\}$. 
	The colors of vertices $u$ that are not $v$ or $w$ are the same with respect to $F$ as with respect to $F'$.
	Observe that if $w$ is green with respect to $F$ if and only if it is green with respect to~$F'$ (as it has the same outgoing edges in $F$ and $F'$).
	If $w$ is red or black with respect to~$F'$, then it is red with respect to~$F$ (as it has an incoming edge and no outgoing edge).
	On the other hand, $v$ could be any color with respect to $F'$, but is necessarily green with resepct to $F$.
	%
	%
	Likewise, we can show the correctness of $s'$.
	Thus, the correct value is stored.

	Let $t$ be a \emph{forget node} with child $t'$ and $Q_t=Q_{t'}\setminus\{v\}$.
	We show that a set $F$ is feasible for $t,R,G,B,s$ if and only if $F$ is also feasible for $t',R,G,B\cup\{v\},s$ or for $t',R,G \cup \{v\},B,s$.
	It then follows that $\DP[t,R,G,B,s]$ stores the correct value.

	To this end, let $F$ be a set of edges that is feasible for $t,R,G,B,s$.
	Since $v \notin Q_t$ and condition (F1) holds, $v$ either has an outgoing edge in $F$ or is a leaf or does not have any incoming edge in $F$. It follows that $v$ is green or black with respect to $F$. From this it is straightforward to confirm that $F$ is feasible for $t',R,G,B\cup\{v\},s$ or for $t',R,G \cup \{v\},B,s$.
	Conversely, if $F$ is feasible for $t',R,G,B\cup\{v\},s$ or for $t',R,G \cup \{v\},B,s$, then $F$ also satisfies condition (F1) for $t,R,G,B,s$ (in particular, the condition is satisfied for $v$ as $v$ is green or black with respect to $F$).
	It is straightforward to confirm that $F$ satisfies the remaining conditions to be feasible with respect to $t,R,G,B,s$.

	Finally, let $t$ be a \emph{join node} with children $t_1$ and $t_2$.
	We show that there is an $F\in \Sstar$ if and only if
	there are partitions $R_1 \cup G_1 \cup B_1$ and $R_2 \cup G_2 \cup B_2$, qualified for $R,G,B$
	and~$s'\le s$ such that
	there are $F_1\in \SstarPrime{t_1,R_1,G_1,B_1,s'}$ and $F_2\in \SstarPrime{t_2,R_2,G_2,B_2,s-s'}$
	with $F_1\cup F_2 = F$.

	First, let $F$ be feasible for $t,R,G,B,s$.
	Let $F_i$ be the subset of edges of $F$ that occur in~$\Net_{t_i}$ for $i\in\{1,2\}$.
	Because every edge is introduced exactly once, the sets $F_1$ and $F_2$ are disjoint.
	We define $s'$ to be the number of edges of $F_1$ that are incident with a leaf.
	Further, we define $R_i$ to be the set of vertices in $Q_t$ that are not leaves in \Net and have an incoming but no outgoing edge in $F_i$ for $i\in \{1,2\}$.
	Similarly, we define $G_i$ to be the set of vertices in $Q_t$ that are leaves or have an outgoing edge in $F_i$, and we define $B_i$ to be the set of vertices in~$Q_t$ that are not leaves and have no incident vertices in $F_i$, for $i\in \{1,2\}$.

	We show that $F_1\in \SstarPrime{t_1,R_1,G_1,B_1,s'}$ and $F_2\in \SstarPrime{t_2,R_2,G_2,B_2,s-s'}$.
	Let $e=(u,v)$ be an edge in $F_i$ such that $v$ is not a leaf and $v\not\in Q_{t_i}$ for~$i\in \{1,2\}$.
	Because $Q_t=Q_{t_i}$, we conclude that $v\not\in Q_{t}$.
	Thus, with (F1) we conclude that~$v$ has an outgoing edge $e^*$ in $F$.
	Because $v$ is not in $Q_t$ and $v$ is in $V_{t_i}$, the vertex $v$ is not in~$V_{t_{3-i}}$.
	Thus, the edge $e^*$ is also in $F_i$ and so $F_i$ satisfies (F1).
	By construction, $F_i$ satisfies (F2) to (F4).
	Also, $F_1$ satisfies (F5), by definition of $s'$.
	As $F_2$ contains the edges of $F$ that are not in $F_1$, there are $s-s'$ edges in $F_2$ that are incident with a leaf.
	Hence, $F_2$ holds (F5) and so, $F_1\in \SstarPrime{t_1,R_1,G_1,B_1,s'}$ and~$F_2\in \SstarPrime{t_2,R_2,G_2,B_2,s-s'}$.
	
	It remains to show that~$R_1\cup G_1 \cup B_1$ and $R_2\cup G_2 \cup B_2$ are qualified for $R,G,B$.
	If $v$ is a leaf, we conclude that $v$ is in $G_1,G_2$ and $G$.
	Further, $v$ has an outgoing edge in $F$ if and only if $v$ has an outgoing edge in $F_1$ or $F_2$.
	We conclude that $G=G_1\cup G_2$.
	For each $v\in R$ there is an incoming edge $e$ but no outgoing edge in $F$.
	Consequently, if $v\in R$ then $v\not\in G_1\cup G_2$.
	Further, $e \in F_i$ for some $i \in \{1,2\}$ but there is no outgoing edge for $v$ in $F_i$.
	Thus, $v\in R_i$ and so~$R\subseteq (R_1\cup R_2) \setminus (G_1 \cup G_2)$.
	If $v\in (R_1\cup R_2) \setminus (G_1 \cup G_2)$ then there is an incoming edge for $v$ in $F_1$ or $F_2$ and therefore in $F$, but no outgoing, and so $v\in R$. 
	We conclude that $R=(R_1\cup R_2) \setminus (G_1 \cup G_2)$.
	From $G=G_1\cup G_2$ and $R=(R_1\cup R_2) \setminus (G_1 \cup G_2)$ we have that $B = B_1 \cap B_2$, and so~$R_1\cup G_1 \cup B_1$ and $R_2\cup G_2 \cup B_2$ are qualified for $R,G,B$, as required.

	For the converse, assume that there are partitions~$R_1\cup G_1 \cup B_1$ and $R_2\cup G_2 \cup B_2$ qualified for $R$ and $G$,
	and $s'\le s$ such that there are $F_1\in \SstarPrime{t_1,R_1,G_1,B_1,s'}$ and $F_2\in \SstarPrime{t_2,R_2,G_2,B_2,s-s'}$.
	We show that $F:=F_1\cup F_2$ is feasible for $t,R,G,B,s$.
	Let $e=(u,v)$ be an edge in $F$ such that~$v$ is not a leaf and $v\not\in Q_t$. 
	Without loss of generality, $e\in F_1$.
	Because $Q_t=Q_{t_1}$, we conclude that $v\not\in Q_{t_1}$.
	Thus, $v$ has an outgoing edge in $F_1$ and therefore in $F$ and $F$ holds (F1).
	Let~$v$ be a vertex of~$Q_t$.
	If~$v\in R=(R_1\cup R_2) \setminus (G_1 \cup G_2)$, then there is an incoming edge $e$ for $v$ in $F_1$ or $F_2$ but~$F_1$ and~$F_2$ do not contain an outgoing edge for $v$. Consequently, $v$ has an incoming edge in~$F$ but no outgoing edge.
	Analogously, we can see that if $v\in G=G_1\cup G_2$ then $v$ has an outgoing edge in $F$ and if $v\in B=B_1\cap B_2$, then~$v$ is not incident with any edge in $F$.
	Because that covers all options, $F$ holds (F2) to (F4).
	Let $\hat E_1$ and $\hat E_2$ be the edges of $F_1$ and $F_2$, respectively, that are incident with a leaf. Thus, $|\hat E_1|=s'$ and $|\hat E_2|=s-s'$.
	Because every edge is introduced exactly once, the sets $F_1$ and $F_2$ and therefore $\hat E_1$ and $\hat E_2$ are disjoint.
	Consequently, $\hat E_1\cup \hat E_2$ contains $s$ edges that are incoming at leaves.
	Thus, $F$ holds (F5).
	Hence, the value that is stored with Equation (\ref{tw:join}) is correct.

	\proofparagraph{Running Time}
	The tree decomposition $T$ has $\Oh(m)$ nodes. Note that comparing and adding numbers takes $\Oh(\log D)$ time, as we may assume we do not deal with numbers larger than~$\Oh(D)$.

	For a leaf node, an introduce vertex nodes and a forget node, the value of each entry~$\DP[t,R,G,B,s]$ can clearly be computed in linear time in $\tw\cdot \log D$.
	Checking the conditions in an introduce edge node requires checking the colors of $v$ and $w$.
	Further, the partition~$R'\cup G'\cup B'$ agrees on all vertices but $v$ and $w$.
	Therefore, we can also compute the values of entries in time linear in $\tw \cdot \log D$ for an introduce edge node.
	Altogether, we can compute the value of all entries of nodes that are not join nodes in~$\Oh(3^\tw \cdot \tw  \cdot k \cdot m  \cdot \log D)$~time.

	For the computation of a join node $t$, we first store $-m\cdot \max_\w-1$ in each entry $\DP[t,R,G,B,s]$ for partitions $R\cup G\cup B$ of $Q_t$ and an integer $s\in [k]_0$.
	Then, iterate over partitions $R_1\cup G_1\cup B_1$ and~$R_2\cup G_2\cup B_2$ of $Q_t$.
	From the partitions~$R_1\cup G_1\cup B_1$ and~$R_2\cup G_2\cup B_2$, compute the implicitly defined partition $R\cup G\cup B$ of $Q_t$ in $\Oh(\tw)$ time. (See \Cref{fig:joinColorings}).
	Iterate over $s\in [k]_0$ and $s'\in [s]_0$.
	If $\DP[t_1,R_1,G_1,B_1,s'] + \DP[t_2,R_2,G_2,B_2,s-s']$ exceeds the value of $\DP[t,R,G,B,s]$, then replace it. Otherwise let $\DP[t,R,G,B,s]$ stay unchanged. 
	After the iterations, we have computed the values of $\DP[t,R,G,B,s]$ for all partitions $R\cup G\cup B$ of $Q_t$ and integer $s\in [k]_0$.
	Therefore, $\Oh(9^{\tw} \cdot \tw \cdot k^2  \cdot \log D)$ time is required to compute all values of entries for each join node $t$.

	Hence, a solution for \MAPPD can be computed in $\Oh(9^\tw \cdot \tw \cdot k^2 \cdot m  \cdot \log D)$ time.
\end{proof}

\section{A Kernelization for Reticulation-Edges}
\label{sec:kernel}
In \Cref{thm:ret}, we presented a branching algorithm that proves that \MAPPD is \FPT when parameterized by the number of reticulations, $\vret$.
In this section, we show that \MAPPD admits a kernelization algorithm of polynomial size with respect to~\eret.
Recall that~\eret is the number of reticulation-edges which need to be removed such that~\Net is a tree.
Observe~$\eret \ge \vret$ and in binary networks they are equal.

We first show how to bound the number of vertices and edges by a polynomial in~\eret, without giving any such bound on the weights of the edges. Afterwards, we will apply a result from \cite{etscheid,frank1987} to get an appropriate bound on the edge weights.

\begin{theorem}
	\label{thm:vkernel}
	Given an instance $\Instance = (\Net, k, D)$ of \MAPPD,
	we can compute an equivalent instance $\Instance^* = (\Net^* = (V^*,E^*,\w^*), k^*, D^*)$ of \MAPPD
	with $|V^*|,|E^*| \in \Oh(\eret^2)$ and $k^* \in \Oh(\eret)$,
	in~$\Oh(m^2 \log^2 m \cdot \log \max_\w)$~time.
\end{theorem}

Throughout this section, assume that $\Instance = (\Net = (V,E,\w), k, D)$ is an instance of \MAPPD with $\rho$ being the root of \Net.
We recall that we do not assume \Net to be binary and we even allow vertices with an in- and out-degree of~1.
Here, we slightly bend our own definitions in which we required that $\w(e) > 0$ for each edge~$e$ and that each vertex either has an in-degree of~1 or an out-degree of~1.
We note that after applying all reduction rules exhaustively, we can adjust the instance to ensure these requirements are met again.
We discuss this step after describing all the reduction rules.

We apply the following reduction rules exhaustively, and each rule is applied only if none of the previous rules apply.
After any of the reduction rules let $\Net' = (V',E',w')$ denote the new network.
We assume that at any application of a reduction rule, all previous rules have been applied exhaustively.

\begin{rr}
	\label{rr:trees}
	Let $v\in V$ be a vertex with children $x$ and~$y$ that are leaves.
	Assume~$\w((v,x)) \ge \w((v,y))$.
	If $v\ne \rho$, then replace the edge~$(v,y)$ with an edge~$(\rho,y)$ of weight~$\w'((\rho,y)) = \w((v,y))$.
\end{rr}
\begin{lemma}
	\label{lem:trees}
	\Cref{rr:trees} is correct and can be applied in~$\Oh(|X| + \log\max_{\w})$~time.
\end{lemma}
\begin{proof}
	\proofparagraph{Correctness}
	We show first that if \Instance is a \yes-instance of \MAPPD then $\Instance' := (\Net',k,D)$ is a \yes-instance of \MAPPD.
	%
	Let $S\subseteq X$ be a solution for \MAPPD on \Instance.
	If~$y\in S$ but~$x\not\in S$, then set $S' := {(S\cup \{x\})} \setminus \{y\}$.
	Then, because $\PD(S') = \PD(S) + \w((v,x)) - \w((v,y)) \ge \PD(S)$, also~$S'$ is a solution.
	Therefore, we may assume without loss of generality that~$x \in S$ or $y \notin S$.
	If $x \notin S$ and $y \notin S$, then observe that $\PD(S) = \PDsub{\Net'}(S)$ since all edges affected by $S$ appear in both networks with the same weight. 
	Similarly, if $x \in S$ and $y \notin S$, then $\PD(S) = \PDsub{\Net'}(S)$.
	Finally, if $x \in S$ and $y \in S$, then any edge $e \in E(\Net)\setminus \{(v,y)\}$ is affected by $S$ in $\Net$ if and only if $e$ is affected by $S$ in $\Net'$ (in particular, if $y$ is an offspring of $e$ in $\Net$ then $x$ is an offspring of $e$ in $\Net'$).
	It follows that  $\PD(S) = \PDsub{\Net'}(S) - \w'((\rho,y)) + \w((v,y)) = \PDsub{\Net'}(S)$.
	Thus, $S$ is a solution for \MAPPD on $\Instance'$ in all cases.

	Conversely, observe that for any edge $e \in E(\Net')\setminus \{\rho y\}$ and any $S \subseteq X$, if $e$ is affected by $S$ in $\Net'$ then it is also affected by $S$ in $\Net$.
	Thus, $\PDsub{\Net'}(S) \le \PD(S)$, and so each solution for \MAPPD on $\Instance'$ is also a solution for \MAPPD on \Instance.
	Thus, if $\Instance'$ is a \yes-instance of \MAPPD, then so is \Instance.

	\proofparagraph{Running Time}
	In~$\Oh(|X|)$ time, we can determine whether there exists a vertex~$v$ which has two children in~$X$, by iterating over~$X$ and marking each taxon's parent.
	If such a vertex~$v$ exists, we only need to compare~$\w((v,x))$ and $\w((v,y))$ in~$\Oh(\log\max_{\w})$~time.
	%
\end{proof}

\looseness=-1
Note that \Cref{rr:trees} might create degree-2-vertices; these are handled by the next reduction rule.
\begin{rr}
	\label{rr:deg-2}
	Let $v\in V$ be a vertex with an in- and out-degree of~1.
	Let $u$ be the parent and $w$ be the child of $v$.
	Remove~$v$ and its incident edges and create an edge~$(u,w)$ of weight~$\w'((u,w)) = \w((u,v)) + \w((v,w))$.
\end{rr}
\begin{lemma}
	\label{lem:deg-2}
	\Cref{rr:deg-2} is correct and can be applied
	in~$\Oh(n + \log\max_{\w})$~time.
\end{lemma}
\begin{proof}
	\proofparagraph{Correctness}
	The correctness follows from the observation that~$\PD(S) = \PDsub{\Net'}(S)$ for any~$S \subseteq X$.
	
	\proofparagraph{Running Time}
	In~$\Oh(n)$ time, we can find out whether there exists a vertex of degree~2. If such a vertex $v$ exists, it takes~$\Oh(\log\max_{\w})$~time to remove $v$ and its incident edges and create the new edge $uw$ with the appropriate weight.
\end{proof}

To evaluate the size of the network after applying reduction rules, we adapt the language of \emph{network generators}.
We refer the reader to~\cite{GambetteGenerators2009, GeneratorsDef2009} for an in-depth study.

\begin{definition}
	\label{def:cor-side}
	
	\begin{itemize}
		\item We use $R$ to denote the set of reticulation vertices of \Net (including vertices of in- and out-degree greater than $1$).
		\item A vertex $v \in V\setminus R$ is a \emph{core-vertex} if $v$ has two children, $u_1$ and $u_2$, where $u_1 \ne u_2$, and $\desc(u_i) \cap R \ne \emptyset$ for each $i \in \{1, 2\}$.
		
		\item We use $Q$ to denote the set of core-vertices of \Net.
		
		\item A \emph{side-path} in $\Net$ is a path from $u$ to $w$ for two vertices $u,w \in R \cup Q$ with no internal vertices in $R\cup Q$.
		The internal vertices of a side-path are \emph{side-vertices}.

		\item We use $P$ to denote the set of side-vertices of \Net.
		
	\end{itemize}
\end{definition}

Note that after exhaustively applying \Cref{rr:trees,rr:deg-2}, every non-leaf vertex is either a reticulation or has at least one reticulation as a descendant.
Therefore, every non-leaf vertex is in $Q$, in $R$, or in~$P$.
Every side-vertex has exactly one child which is a leaf.

The following result is similar to one in~\cite{GambetteGenerators2009}. For completeness we prove it here.
\begin{observation}\label{obs:core-retic-bound}
	Every network~\Net has $\Oh(\eret)$ side-paths, and $|R|+|Q| \in \Oh(\eret)$
\end{observation}
\begin{proof}
	Observe that $|R| = \vret \leq \eret$. 
	
	To show the other two claims, we give a bound on the number of side-paths from above and below.
	As every side-path in $\Net$ ends at a core vertex or reticulation, we have that the number of side-paths is at most 
	$$\sum_{v \in Q\cup R} \deg^-(v)
	= \sum_{r \in  R} (\deg^-(r) - 1) + |Q| + |R| 
	= \eret +  |Q| + |R| \leq 2\cdot \eret + |Q|.$$
	To give a bound from below, consider that each core-vertex has at least two side-paths leaving it, from which it follows that the number of side-paths is at least $2|Q|$.
	It follows that $2|Q| \leq  2\cdot \eret + |Q|$ and so $|Q| \leq 2\eret$.
	This in turn implies that the number of side-paths is at most $4\cdot \eret$.
\end{proof}

\begin{rr}
	\label{rr:a>a+b}
	Let~$v$ and~$w$ be side-vertices and let~$v$ be the parent of~$w$.
	Further, let~$x_v$ and~$x_w$ be the only leaf-children of $v$ and~$w$, respectively.
	If~$\w((v,x_v)) \le \w((v,w)) + \w((w,x_w))$, then replace the edge~$(v,x_v)$ with an edge~$(\rho,x_v)$ of weight~$\w'((\rho,x_v)) = \w((v,x_v))$.
\end{rr}
\begin{lemma}
	\label{lem:a>a+b}
 	\Cref{rr:a>a+b} is correct and can be applied
	in~$\Oh(n \cdot \log\max_{\w})$~time.
\end{lemma}
\begin{proof}
	\proofparagraph{Correctness}
	The proof follows similar lines as the proof of \Cref{lem:trees}.
%
%
%
%
	Let $S\subseteq X$ be a solution for \MAPPD on $\Instance$.
	Assume that $x_v \in S$ but $S \cap \off(w) = \emptyset$.
	Then, define  $S' := {(S\cup \{x_w\})} \setminus \{x_v\}$.
	Observe that $(w,x_w)$ and $(v,w)$ are affected by~$S'$ but not by~$S$, while the only edge affected by~$S$ and not by~$S'$ is~$(v,x_v)$.
	Thus $\PD(S') = \PD(S) - \w((v,x_v)) + \w((w,x_w)) + \w((v,w))$, which, by the condition of the reduction rule, is at least $\PD(S)$.
	Therefore, also~$S'$ is a solution for \MAPPD on \Instance with $x_v\not\in S'$.
	Thus, we may now assume that $S$ contains a taxon from $\off(w)$ or $x_v\not\in S$. 
	In either case we can observe that $\PD(S) = \PDsub{\Net'}(S)$.

	Conversely, observe that any edge $e \in E(\Net')\setminus\{(\rho,x_v)\}$ affected by $S$ in $\Net'$ is also affected by $S$ in $\Net$.
	Thus, $\PDsub{\Net'}(S) \le \PD(S)$ for any $S \subseteq X$, and so each solution for \MAPPD on~$\Instance'$ is also a solution for \MAPPD on \Instance.
	
	\proofparagraph{Running Time}
	In~$\Oh(n \cdot \log\max_{\w})$ time, we can find out whether there exists a vertex satisfying the conditions of~$v$ in \Cref{rr:a>a+b}, and if so edit the network accordingly.
\end{proof}

\begin{rr}
	\label{rr:strings}

	Let~$u$ and~$w$ be core-vertices or reticulations with $u \neq \rho$ and
	for some~$\ell > 1$
	let~$v_1,\dots,v_\ell$ be side-vertices such that $(u,v_1), (v_\ell,w), (v_i,v_{i+1}) \in E$ for each $i\in [\ell-1]$.
	Remove the edge~$(v_1,v_2)$ and add edges~$(\rho,v_2)$ and $(v_1,w)$ of respective weight $\w'((\rho,v_2)) = \w((v_1,v_2))$ and $\w'((v_1,w)) = 0$.
\end{rr}
An application of this reduction rule is depicted in Figure~\ref{fig:strings}.

\begin{figure}[t]
	\centering
	\begin{tikzpicture}[scale=0.7,every node/.style={scale=0.7}]
		\node (u) {$u$}
		child {node {} edge from parent[draw=none]}
		child {node {$v_1$} edge from parent[-{Stealth[length=6pt]}]
			child {node {$x_1$} edge from parent[-{Stealth[length=6pt]}]}
			child {node {$v_2$} edge from parent[-{Stealth[length=6pt]}]
				child {node {$x_2$} edge from parent[-{Stealth[length=6pt]}]}
				child {node {$v_\ell$} edge from parent[dotted,-{Stealth[length=.1pt]}]
					child {node {$x_\ell$} edge from parent[solid,-{Stealth[length=6pt]}]}
					child {node {$w$} edge from parent[solid,-{Stealth[length=6pt]}]};
				};
			};
		};
		
		\node (root) at (2,0) {$\rho$};
		\draw[bend right, dashed,-{Stealth[length=6pt]}] (root) to (u);
		
		\node at (3,0) {(1)};
		\draw (4,.25) to (4,-6.25);
	\end{tikzpicture}
	\begin{tikzpicture}[scale=0.7,every node/.style={scale=0.7}]
		\node (root) {$\rho$}
		child {node (v2) {$v_2$} edge from parent[-{Stealth[length=6pt]}]
			child {node {$v_3$} edge from parent[-{Stealth[length=6pt]}]
				child {node {$v_\ell$} edge from parent[dotted,-{Stealth[length=.1pt]}]
					child {node (w) {$w$} edge from parent[solid,-{Stealth[length=6pt]}]}
					child {node {$x_\ell$} edge from parent[solid,-{Stealth[length=6pt]}]}
				}
				child {node {$x_3$} edge from parent[-{Stealth[length=6pt]}]}
			}
			child {node {$x_2$} edge from parent[-{Stealth[length=6pt]}]}
		};
		\node[left of=v2, xshift=-35mm] (u) {$u$}
		child {node {} edge from parent[draw=none]}
		child {node (v1) {$v_1$} edge from parent[-{Stealth[length=6pt]}]
			child {node {$x_1$} edge from parent[-{Stealth[length=6pt]}]}
			child {node {} edge from parent[draw=none]}
		};
		\draw[-{Stealth[length=6pt]}] (v1) to (w);
		\draw[bend right, dashed,-{Stealth[length=6pt]}] (root) to (u);
		
		\node at (-.75,-.5) {$\w((u,v_1))$};
		\node at (-3.25,-4.75) {$0$};
		
		\node at (1,0) {(2)};
	\end{tikzpicture}
	\caption{This figure in (1) depicts the path from a core-vertex~$v$ to another core vertex~$w$ and in~(2) an application of \Cref{rr:strings} to the path depicted in (1).}
	\label{fig:strings}
\end{figure}
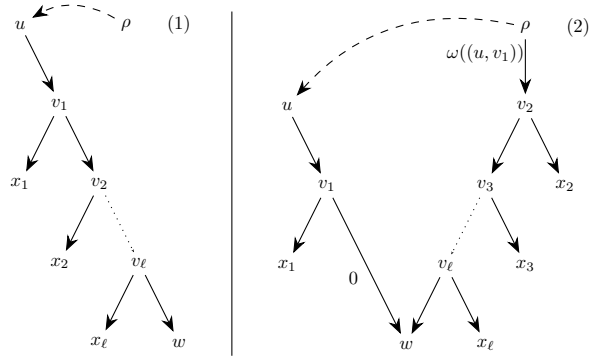%

\begin{lemma}
	\label{lem:strings}
	\Cref{rr:strings} is correct and can be applied 
	in $\Oh(m + \log\max_\w)$ time.
\end{lemma}
\begin{proof}
	\proofparagraph{Correctness}
	Observe that any edge $e \in E(\Net')\setminus\{(\rho,v_2), (v_1,w)\}$ that is affected by~$S \subseteq X$ in $\Net'$ is also affected by~$S$ in $\Net$.
	Furthermore, $(\rho,v_2)$ is affected by $S$ in $\Net'$ if and only if $(v_1,v_2)$ is affected by $S$ in $\Net$.
	Therefore, $\PDsub{\Net'}(S) \le \PD(S) + \w'((v_1,w)) = \PD(S)$ for every set~$S \subseteq X$, and each solution for \MAPPD on~$\Instance'$ is also a solution for \MAPPD on \Instance.
	
	It remains to show that 
	if $\Instance$ is a \yes-instance of \MAPPD, then so is $\Instance'$.
	We first observe some facts.
	Let $x_i$ denote the leaf child of $v_i$ for each $i \in [\ell]$.
	Note that due to \Cref{rr:a>a+b} we may assume that $\w((v_i,x_i)) \ge \w((v_i,v_{i+1})) + \w((v_{i+1},x_{i+1}))$ for each~$i \in [\ell - 1]$.
	Consequently,
	\begin{eqnarray}
		\w((v_1,x_1))
		&\geq& \w((v_2,x_2)) + \w((v_1,v_2))\\
		&\geq& \w((v_3,x_3)) + \w((v_2,v_3)) + \w((v_1,v_2))\\
		&\geq& \w((v_i,x_i)) + \sum_{j = 1}^{i-1}\w((v_j,v_{j+1}))
	\end{eqnarray}
	for each $i \in [\ell - 1]$. 
	
	Let $S$ be a solution for \MAPPD on \Instance.
	Observe that if $S$ contains $x_1$ or an offspring of~$w$, then~$\PD(S) = \PDsub{\Net'}(S)$.
	Similarly, if $S$ does not contain a taxon in~$x_1, \dots, x_\ell$ then~$\PD(S) = \PDsub{\Net'}(S)$.
	So assume $S$ does not contain $x_1$ nor an offspring of~$w$ but~$x_i \in S$ for some~$i\in [\ell]$ with~$i>1$.
	Define~$S':= (S \cup \{x_1\}) \setminus \{x_i\}$.
	Then, $\PD(S') \ge \PD(S) + \w((v_1,x_1)) - \w((v_i,x_i)) - \sum_{j = 1}^{i-1}\w((v_j,v_{j+1}))$, which by the observation above is at least~$\PD(S)$.
	As $x_1 \in S'$ we have $ \PDsub{\Net'}(S') = \PD(S')$, and so $S'$ is  a solution for \MAPPD on $\Instance'$.

	\proofparagraph{Running Time}
	To find the vertices~$u$ and~$w$ we iterate over the vertices in~$Q$ as vertex~$u$ and consider each outgoing edge as a path to~$u$.
	Therefore, all possible combinations of~$u$ and~$w$ are found in~$\Oh(m)$~time.
	Each operation can be executed in  $\Oh( \log\max_\w)$ time.
\end{proof}

We now categorize the vertices in some sets to be able to easier refer to them.

\begin{definition}
	\label{def:top-strings}
	\begin{itemize}
		\item Define $A$ to be the set of leaf-children of $\rho$.
		\item Define $B$ to be the set of descendants of any vertex in $R \cup (Q \setminus \{\rho\})$.
		\item Define $Y$ to be the side-vertices which are children of $\rho$ and define $Y_X$ to be the leaf-children of $Y$.
		\item Define $Z$ to be the side-vetices which are not in $B\cup Y$---that is, vertices that are in a side which is outgoing of the root, but not the top vertex of that side---and define $Z_X$ to be the leaf-children of $Z$.
		\item Define $x^* := \arg\max_{x\in A} \w((\rho,x))$ and~$a^* := \w((\rho,x^*))$.
		\item Define $y^* := \arg\max_{y\in Y_X} \w((\rho,v_y)) + \w((v_y,y))$ and~$c^* := \w((v_{y^*},y^*))$, where $v_y$ is the parent of $y$.
	\end{itemize}
\end{definition}

\Cref{rr:trees,rr:a>a+b,rr:strings} have only transformed the instance and not made it smaller. We are now ready to start removing vertices.
In the following, for a taxon~$x\in X$, we use an operation~\emph{remove~$x$} when we delete~$x$ from~$X$ and~$V$ as well as the incoming edge of~$x$ from~$E$.
Further, we use an operation \emph{save~$x$} consisting of these steps:
\begin{itemize}
	\item Reduce $k$ by $1$;
	\item Reduce $D$ by $\PD(\{x\})$;
	\item Delete~$x$ from~$X$;
	\item Delete all edges affected by~$x$ (denoted with~$E_{\{x\}}$); and
	\item Identify all vertices with no incoming edge to a single root.
\end{itemize}

\begin{lemma}
	\label{lem:save-remove}
	For a given instance \Instance of \MAPPD and a taxon~$x$,
	\begin{itemize}
		\item \Instance has a solution~$S$ with~$x \in S$ if and only if~$S\setminus \{x\}$ is a solution for the instance after~$x$ is saved; and
		\item \Instance has a solution~$S$ with~$x \notin S$ if and only if~$S$ is a solution for the instance after~$x$ is removed.
	\end{itemize}
\end{lemma}
\begin{proof}
	To see the first claim, let $\Net'$ be the network after saving $x$ and let $S \subseteq X$ contain~$x$.
	Any edge in $\Net'$ affected by $S\setminus\{x\}$ has a corresponding edge in $E(\Net)\setminus E_{\{x\}}$ with the same offspring, from which it follows that $\PD(S) \geq \PDsub{\Net'}(S\setminus\{x\}) + \PD(\{x\})$.
	
	Conversely, if $x \in S \subseteq X$ then any edge in $E_S \setminus E_{\{x\}}$ has a corresponding edge in $\Net'$ affected by $S \setminus \{x\}$, from which it follows that $\PDsub{\Net'}(S\setminus \{x\}) \geq \PD(S) - \PD(\{x\})$.
	
	The second claim follows immediately from the definition of $\PD(S)$.
\end{proof}

\begin{rr}
	\label{rr:top-strings}
	If $k > |B| + |Y|$ and~$a^* > c^*$, then save~$x^*$.
	If $k > |B| + |Y|$ and~$a^* \le c^*$, then save $y^*$.
\end{rr}
\begin{lemma}
	\label{lem:top-strings}
	Reduction Rule~\ref{rr:top-strings} is correct and can be applied 
	in $\Oh(m \log\max_\w)$ time.
\end{lemma}
\begin{proof}
	\proofparagraph{Correctness}
	Let $v_{1,1}, v_{2,1}, \dots v_{|Y|,1}$ denote the vertices of $Y$.
	Let~$v_{j,1} v_{j,2},\dots v_{j, \ell_j}$ be the side-vertices on the path from $v_{j,1}$ to a core vertex for each $j \in [|Y|]$.
	Let~$z_{j,i}$ be the leaf child of $v_{j,i}$ for each $j \in [|Y|]$, and~$i \in [\ell_j]$.
	This mapping is unique after~\Cref{rr:trees,rr:deg-2} have been applied exhaustively.
	Observe that $Z$ is the set~$\{v_{j,i} \mid 2 \leq j \leq |Y|, i \in [\ell_j]\}$.
	
	Due to~\Cref{rr:a>a+b}, we have $\w((v_{j,1},y_{j,1})) > \w((v_{j,i},y_{j,i})) + \sum_{h=1}^{i-1}\w((v_{j,h},v_{j,h+1}))$ for each $j \in [|Y|]$, and each~$i \in [\ell_j]$.
	Consequently, we may assume that if $y_{j,i}$ is in a solution for some~$i > 1$, then so is~$y_{j,1}$.
	
	Furthermore, for any solution that contains $y_{j,i}$ and $y_{j,1}$ for some $i>1$, we can assume that the solution contains~$y^*$ as otherwise replacing $y_{j,i}$ with $y^*$ gives another solution, because
	\begin{eqnarray}
		&& \w((\rho,v_{y^*}))+ \w((v_{y^*},y^*))\\
		&\geq& \w((\rho,v_{j,1})) + \w((v_{j,1},y_{j,1}))\\
		&>& \w((\rho,v_{j,1})) + \sum_{h=1}^{i-1} \w((v_{j,h},v_{j,h+1})) + \w((v_{j,i},y_{j,i}))
	\end{eqnarray}
	where~$v_{y^*}$ is the parent of~$y^*$.
	%
	Thus, if a solution~$S$ contains any element of~$Z_X$ we can assume~$S$ also contains $y^*$.
	
	Now, suppose $k > |B| + |Y|$.
	Then, any solution~$S$ contains at least one element of~$A \cup Z_X$.
	This implies that $S$ contains at least one taxon of $A \cup \{y^*\}$ as~$S$ contains~$y^*$ if it contains any taxon in $Z_X$.
	If $a^* > c^*$, then $S$ contains $x^*$, as otherwise we could replace a taxon from~$(A\setminus \{x^*\}) \cup \{y^*\}$ with $x^*$.
	So, in this case, $S$ contains $x^*$, and we can save $x^*$ by~Lemma~\ref{lem:save-remove}.
	Otherwise, we may assume $S$ contains $y^*$, as otherwise we can replace an element from $A$ with $y^*$, and we can save $y^*$ by~Lemma~\ref{lem:save-remove}.

	\proofparagraph{Running Time}
	We can compute the size of~$B$ and~$Y$ and find $a^*$ and~$c^*$ in $\Oh(n \log\max_\w)$ time.
	Saving $x^*$ or $y^*$ takes $\Oh(m \log\max_\w)$ time.
\end{proof}

\begin{rr}
	\label{rr:A}
	Let $x_1,\dots,x_{|A|}$ be an ordering of the taxa in $A$ such that $\w((\rho,x_i)) \ge \w((\rho,x_{i+1}))$ for each $i\in [|A|-1]$.
	If $k > |A|$, then remove $x_{k+1},\dots,x_{|A|}$ and their incident edges from~$\Net$.
\end{rr}
\begin{lemma}
	\label{lem:A}
	Reduction Rule~\ref{rr:A} is correct and can be applied 
	in $\Oh(n\log n \cdot \log\max_{\w})$~time.
\end{lemma}
\begin{proof}
	\proofparagraph{Correctness}
	Clearly, any solution for \MAPPD on instance~$\Instance'$ is also a solution for \MAPPD on~\Instance.
	Conversely, let~$S$ be a solution for \MAPPD on~$\Instance$.
	If $S \cap \{x_{k+1},\dots,x_{|A|}\} = \emptyset$, then $S$ is also a solution for \MAPPD on~$\Instance'$.
	Assume that $x_i \in S$ for some~$i\in \{k+1, \dots, |A|\}$.
	As~$|S| \le k$ we conclude that there is a taxon~$x_j$ for $j\in [k]$ with $x_j \not\in S$.
	Because~$\w((\rho,x_j)) \ge \w((\rho,x_i))$, we conclude that $(S \cup \{x_j\}) \setminus \{x_i\}$ is also a solution for \MAPPD on~\Instance.
	Thus, we may assume that~$S \cap \{x_{k+1},\dots,x_{|A|}\} = \emptyset$.

	\proofparagraph{Running Time}
	We can sort $A$ in $\Oh(n\log n \cdot \log\max_{\w})$~time.
\end{proof}

\begin{rr}
	\label{rr:pathlength}
	
	Let $z_0,\dots,z_{k+1}\in Y\cup Z$ be vertices with $z_{i}$ being the parent of $z_{i+1}$ for $i\in [k]_0$.
	Remove $x_k$, the leaf-child of~$z_k$ (and apply \Cref{rr:deg-2} afterwards).
	%
\end{rr}
\begin{lemma}
	\label{lem:pathlength}
	\Cref{rr:pathlength} is correct and can be applied exhaustively in $\Oh(n + \log\max_\w)$~time.
\end{lemma}
\begin{proof}
	\proofparagraph{Correctness}
	Because we applied~\Cref{rr:trees,rr:deg-2,rr:a>a+b} exhaustively, each vertex $z_i$ has exactly one leaf-child~$x_i$ and~$\w((z_i,x_i)) > \w((z_{i+1},x_{i+1})) + \w((z_i,z_{i+1}))$.
	Consequently, we can assume that if~$x_i$ is in a solution then so are~$x_0, \dots, x_{i-1}$.
	Therefore, $x_{k}$ cannot be in a solution of size~$k$, and applying~\Cref{rr:pathlength} is correct.
	
	\proofparagraph{Running Time}
	Similar to the proof of~\Cref{lem:deg-2}, we can find an appropriate vertex $z_k$ in~$\Oh(n)$~time and edit the network in $\Oh(\log\max_\w)$ time.
\end{proof}

Finally, we have everything to proof this section's main theorem.
\begin{proof}[Proof of Theorem~\ref{thm:vkernel}]
	For a given instance~$\Instance = (\Net = (V,E), k, D)$ apply all of the reduction rules exhaustively to receive an instance~$\Instance' = (\Net' = (V',' k', D')$ of \MAPPD in which none of the reduction rules apply.
	As discussed above, this instance may have vertices with in- and out-degree larger than $1$, and edges with weight $0$, neither of which are formally allowed by our definition of \MAPPD. We now discuss how the instance may be edited to fix these issues.

	First, we replace each vertex~$v$ which has an in-degree and an out-degree of larger than~1 with a vertex~$v_{\text{in}}$ receiving all the incoming edges of~$v$ and a vertex~$v_{\text{out}}$ receiving all the outgoing edges, and add an edge~$(v_{\text{in}},v_{\text{out}})$ of weight~0.
	Finally, we multiply each edge weight and $D'$ by $(|E'|+1)$, and the edges with weight~0 to weight~1.
	Let $\Instance^* = (\Net^* = (V^*,E^*), k^*, D^*)$ denote the resulting instance of \MAPPD.

	In what follows we let $R,Q,P$ and $A,B,Y,Z$ denote the sets of vertices in $\Net$ as defined in \Cref{def:cor-side,def:top-strings} respectively. We let $R^*,Q^*,P^*$ and $A^*,B^*,Y^*,Z^*$ denote the corresponding vertices in $\Net^*$.
	
	The correctness follows from the correctness of the reticulation rules.
	
	\proofparagraph{Running time}
	Observe that~\Cref{rr:top-strings,rr:A,rr:pathlength} remove one taxon each, while~\Cref{rr:trees,rr:a>a+b} decrease~$|X\setminus A|$ by $1$.
	As none of the reduction rules increase $|X|$ or $|X\setminus A|$, these rules are applied at most $|X| + |X\setminus A| \leq 2 |X|$ times in total.
	For~\Cref{rr:deg-2}, we note that~\Cref{rr:trees,rr:a>a+b,rr:top-strings,rr:A,rr:pathlength} each create one additional degree-2 vertex, and~\Cref{rr:strings} does not create any.
	Thus, every application of~\Cref{rr:deg-2} is either on at most~$n$ initial degree-2-vertices or occurs after an application of one of \Cref{rr:trees,rr:a>a+b,rr:top-strings,rr:A,rr:pathlength}, and thus~\Cref{rr:deg-2} is applied at most~$2 \cdot |X| + n$ times.

	Each application of \Cref{rr:strings} reduces the number of side-paths starting at a non-root vertex of length of at least~2.
	None of the reduction rules increase this measure, and so the number of application of~\Cref{rr:strings} is bounded by the number of side-paths in the original network, which is $\Oh(\eret)$ by~\Cref{obs:core-retic-bound}.
	
	Therefore, the total number of applications of all reduction rules is $\Oh(|X| + n + \eret) \ \Oh(n)$. 
	Note that a single application of any reduction rule takes $\Oh((m+ n\log n)\cdot \log\max_\w)$ time by~\Cref{lem:trees,lem:deg-2,lem:a>a+b,lem:strings,lem:top-strings,lem:A,lem:pathlength}.	
	An application of~\Cref{rr:deg-2} may increase $\max_\w$, but the total weight of the network never increases and therefore the maximum weight of any edge at any stage in the algorithm is at most the total weight of the network, or at most $m \cdot \max_\w$.
	Therefore, the maximum running time of any application of a reduction rule is $\Oh((m+ n\log n)\cdot (\log m + \log\max_\w))$.
	We conclude that applying all rules exhaustively takes $\Oh((|X| + \eret) \cdot (m+ n\log n)\cdot (\log m +  \log\max_\w))$ time,
    which can be summarized as~$\Oh(m^2 \log^2 m \cdot \log \max_\w)$.
    The final step of replacing vertices with in-and out-degree greater than $1$ and adjusting edge-weights can be done in $\Oh(m \cdot (\log m + \log \max_\w))$, and so the final running time remains~$\Oh(m^2 \log^2 m \cdot \log \max_\w)$.

	\proofparagraph{Size}
	Each application of~\Cref{rr:strings} increases the number of reticulation-edges by one.
	Therefore, note that by \eret we refer to the parameter in the original network.
	
	Observe that none of the reduction rules (except for~\Cref{rr:strings})
	change the number of reticulations or core vertices in the network.
	Reduction Rule \ref{rr:strings} may turn some core vertices into reticulations but otherwise does not create new core vertices or reticulations.
	Therefore, we have $|R^*| + |Q^*| = |R| + |Q| \in \Oh(\eret)$ by Observation~\ref{obs:core-retic-bound}.
	
	As~\Cref{rr:strings} is exhaustively applied, we have that each side-path in~$\Net^*$ not leaving the root has at most one internal vertex, and this vertex has one leaf child. 
	There are $\Oh(\eret)$ side-paths in $\Net^*$ by~\Cref{obs:core-retic-bound} and the fact that none of the reduction rules increase the number of side-paths not leaving the root.
	Thus, we have that the total number of vertices reachable from $R^* \cup (Q^*\setminus \rho)$ is at most $\Oh(|R^*| + |Q^*|) + \Oh(\eret) = \Oh(\eret)$.
	That is, $|B^*| \in \Oh(\eret)$.
	
	The size of~$Y^*$ is the number of side-paths starting at the root.
	There are at most~$|Q| + 2\cdot\eret$ such paths in the original instance.
	Each application of~\Cref{rr:strings} adds one such path.
	We saw that~$|Q| \in \Oh(\eret)$ and~\Cref{rr:strings} can be applied~$\Oh(\eret)$ times.
	We conclude~$|Y^*| \in \Oh(\eret)$.
	
	After~\Cref{rr:top-strings} has been applied exhaustively, we may conclude that~$k^* \le |B^*| + |Y^*| \in \Oh(\eret)$.
	After~\Cref{rr:A} has been applied exhaustively, there are at most~$k^*$ vertices in~$A^*$
	and 
	after~\Cref{rr:pathlength} has been applied exhaustively, each side-path has length at most~$k^*$ such that~$|Y^*| + |Z^*| \le |Y^*| \cdot (k^* + 1) \in \Oh(\eret^2)$.
	We conclude~$|V^*| \in \Oh(\eret^2)$.
	
	We have~$\eretwithoutN_{\Net^*} \le \eret + |Q^*| + |R^*| \in \Oh(\eret)$. Because~$\eret = |E| - |V|$ in any network, we conclude
	$|E^*| = \eretwithoutN_{\Net^*} + |V^*| \in \Oh(\eret^2)$.
\end{proof}

From Theorem~\ref{thm:vkernel}, we have that in polynomial time we can reduce any instance~\Instance of \MAPPD to an equivalent instance $\Instance^* = (\Net^* = (V^*,E^*,{\w^*}), k^*, D^*)$ in which~$|V^*|$, $|E^*|$, and~$k^*$ are all bounded by a polynomial in $\eret$.
This does not guarantee a polynomial kernel, as the encoding size of $D^*$ or $\max_{\w^*}$
could be much larger than~$|V^*|$ or~$|E^*|$.
Fortunately, we can apply a result of \cite{etscheid,frank1987} to bound these values, as follows.

Let~$e_1,\dots,e_{m^*}$ be an order of the edges after applying all reduction rules.
We define~$w_i := \w^*( e_i )$ for each~$i\in [m^*]$ and~$W := D^*$.
In polynomial time, we can compute positive numbers~$w_1', \dots, w_{m^*}'$, and~$W'$ such that the total encoding length is~$\Oh( (m^*)^4 ) \in \Oh( \eret^8 )$ with~$\sum_{i \in S} w_i \ge W$ if and only if~$\sum_{i \in S} w_i' \ge W'$ for every~$S \subseteq [m^*]$ by~\cite[Corollary~2]{etscheid}.

We directly conclude the following.

\begin{theorem}
	\label{thm:kernel}
	\MAPPD admits a polynomial size kernelization algorithm for the number of reticulation-edges~\eret.
\end{theorem}

\section{Discussion}
\label{sec:discussion}
In this paper, we considered \MAPPDlong, a problem in maximizing phylogenetic diversity in networks, and analyzed it within the framework of parameterized complexity.
We showed that \MAPPD is \Wh 2-hard parameterized by $k$, the size of the solution.
We further were able to show an equivalence between \MAPPD parameterized by $k$ and \wpSClong parameterized by the size of the solution.
Thus, establishing the exact complexity class of \wpSClong, which seems to be of interest, would also establish the exact complexity class of \MAPPD.
On the positive side, we showed that \MAPPD is \FPT when parameterized with the number of reticulations~$\vret$ and with respect to the treewidth~$\tw$ of the network.
We further showed that \MAPPD admits a kernelization algorithm with respect to the number of reticulation-edges~$\eret$.

Because $\vret$ is smaller than~$\eret$, it is natural to ask if the kernelization result can also be shown for $\vret$.
We also ask whether \MAPPD parameterized by $k$ is \Wh 2-complete.

In the bigger picture, we have not addressed the question of finding a suitable measure for phylogentic diversity in networks.
The all-paths phylogenetic diversity measure $\PD$ considered in this paper is one of four measures considered in~\cite{bordewich}. The second measure, which they call \emph{Network-PD}, requires not only weights on each of the edges in the network, but also an \emph{inheritance proportion} $p(e)$ on each edge $e = (u,v)$ leading into a reticulation. This value denotes the expected number of features that are expected to be passed from $u$ to $v$. Network-PD is a generalization of AllPaths-PD, as the measures are equivalent when all inheritance proportions are $1$.
The authors also consider two additional measures, \emph{MinWeightTree-PD} and \emph{MaxWeightTree-PD}, in which the respective minimum and maximum displayed tree is considered.
\emph{Average-Tree PD}, which uses inheritance proportions to create a weighted average over all displayed trees, has been studied in~\cite{AVGTree}. This measure is bounded from below and above by MinWeightTree-PD and MaxWeightTree-PD respectively. See also~\cite{vanIersel2026averagetree}.

Since the first publication of this paper, Network-PD has been examined from a parameterized point of view and it has been proven that the problem of maximizing Network-PD is also \FPT when parameterized by $\vret$ on binary networks, but is already NP-hard on level-1 networks~\cite{MaxNPD}.
We note that, since Network-PD generalizes AllPaths-PD, our hardness results for $k$ and $\kbar$ also carry across to Network-PD.
The measure Average-Tree-PD is \#P-hard, but \FPT when parameterized by the network's level on networks of bounded degree~\cite{AVGTree}.
Recently, it has been shown that optimizing MaxWeightTree-PD and computing values for MinWeightTree-PD are both \FPT with respect to node scanwidth~\cite{holtgrefe2026tractable}. It remains open whether MinWeightTree-PD can also be optimized efficiently.

Besides computational issues, there is also the important question of which definition of phylogenetic diversity provides the most realistic measure of actual biodiversity. Any useful definition must strike a balance of being as biologically realistic as possible while still being computationally feasible to work with.

\bibliography{ref}
\bibliographystyle{plainurl}   

\end{document}